\documentclass[conference]{IEEEtran}
\IEEEoverridecommandlockouts

\usepackage{cite}
\usepackage{booktabs}
\usepackage{amsmath,amssymb,amsfonts}
\usepackage{amsthm}
\usepackage{algorithm}
\usepackage{algorithmic}
\usepackage{graphicx}
\usepackage{textcomp}
\usepackage{xcolor}
\newtheorem{proposition}{Proposition}
\newtheorem{theorem}{Theorem}
\newtheorem{lemma}{Lemma}
\newtheorem{corollary}{Corollary}

\usepackage{hyperref}
\newtheorem{remark}{Remark}
\allowdisplaybreaks
\newtheorem{definition}{Definition}
\def\BibTeX{{\rm B\kern-.05em{\sc i\kern-.025em b}\kern-.08em
    T\kern-.1667em\lower.7ex\hbox{E}\kern-.125emX}}
\begin{document}

\title{Privacy Enhancement in Over-the-Air Federated Learning via Adaptive Receive Scaling\\
}

\author{Faeze Moradi Kalarde$^{\star}$, Ben Liang$^{\star}$, Min Dong$^{\dagger}$, Yahia A. Eldemerdash Ahmed$^{\ddagger}$, and Ho Ting Cheng$^{\ddagger}$ \\
$^{\star}$Department of Electrical and Computer Engineering, University of Toronto, Canada \\
$^{\dagger}$Department of Electrical, Computer and Software Engineering, Ontario Tech University, Canada \\
$^{\ddagger}$Ericsson Canada, Canada
\thanks{This work was supported in part by Ericsson, the Natural Sciences and
Engineering Research Council of Canada, and Mitacs.}
}

\maketitle

\begin{abstract}
In Federated Learning (FL) with over-the-air aggregation, the quality of the signal received at the server critically depends on the receive scaling factors. While a larger scaling factor can reduce the effective noise power and improve training performance, it also compromises the privacy of devices by reducing uncertainty. In this work, we aim to adaptively design the receive scaling factors across training rounds to balance the trade-off between training convergence and privacy in an FL system under dynamic channel conditions. We formulate a stochastic optimization problem that minimizes the overall Rényi differential privacy (RDP) leakage over the entire training process, subject to a long-term constraint that ensures convergence of the global loss function. 
Our problem depends on unknown future information, and we observe that standard Lyapunov optimization is not applicable. Thus, we develop a new online algorithm, termed AdaScale, based on a sequence of novel per-round problems that can be solved efficiently.
We further derive upper bounds on the dynamic regret and constraint violation of AdaSacle, establishing that it achieves diminishing dynamic regret in terms of time-averaged RDP leakage while ensuring convergence of FL training to a stationary point. 
Numerical experiments on canonical classification tasks show that our approach effectively reduces RDP and DP leakages compared with state-of-the-art benchmarks without compromising learning performance.
\end{abstract}

\section{Introduction}
\label{INTRODUCTION}
Federated Learning (FL) leverages the computational capabilities of edge devices by allowing them to collaboratively train a global model on their local data without requiring data to be shared \cite{mcmahan2017communication}. 
To enable efficient uplink transmission from devices to the server, over-the-air (OTA) aggregation via analog transmission has emerged as an effective solution \cite{B1, firstyang2020federated, yang2020federated, B2, B4, B3, B5, B6, B7}. In each training round of OTA FL, the devices simultaneously transmit their local signals using analog modulation over a shared multiple access channel, enabling natural model aggregation through signal superposition. 

Nevertheless, OTA computation is susceptible to aggregation errors introduced by receiver noise and channel distortion.
To mitigate these errors, in each training round of OTA FL, each device scales its local signal by a transmit weight, while the server applies a scaling factor to the received signal. The design of the receive scaling factors over training rounds significantly affects the quality of the aggregated signal and thus influences training convergence.

Another concern in FL is privacy leakage, as the local signals sent from the devices reveal information about their underlying data \cite{Digitalwork1, Digitalwork3, Digitalwork2, Digitalwork4, DigitalWork5}.
To reduce data privacy risks, differential privacy (DP) \cite{dwork2014algorithmic} is commonly employed in FL. In the standard DP framework, each device clips its per-sample gradients and adds artificial noise to the batch-averaged gradients before transmission. However, in OTA FL, adding artificial noise is not necessary, as the inherent receiver noise serves as privacy noise and can provide the desired level of privacy \cite{ 18, 20, perRoundPrivacy, 19, 21, 27, 28}. Nevertheless, the receive scaling factors determine the effective noise power at the server and, consequently, the level of privacy leakage. 
Thus, it is essential to design the receive scaling factors to balance the trade-off between training convergence and privacy.

There are two main challenges in designing the receive scaling factors. First, while privacy leakage occurs in each training round, the overall leakage over all rounds is our ultimate concern. Existing works either ignore the overall leakage \cite{18, 20, perRoundPrivacy}, or use the Advanced Composition Theorem for DP over all rounds \cite{19, 21, 28, 27}, which is known to be a loose approximation, especially over a large number of rounds \cite{dwork2014algorithmic, rdp1, 23}. Second, the receive scaling decisions are coupled over the training rounds by the overall privacy leakage and the FL convergence objectives, while the future communication channel state is usually unknown. This necessitates an \textit{online} algorithm to design receive scaling over time. Existing solutions either depend on simplified assumptions about channel conditions \cite{19, 21, 27} or are heuristic-based \cite{28}, lacking theoretical performance guarantees to assess how closely they approximate the optimal solution (see Section~\ref{RELATEDWORKS}).

In this work, we aim to adaptively design the receive scaling factors for an OTA FL system under time-varying channel conditions. We address the aforementioned challenges, first, by employing the Rényi differential privacy (RDP) framework \cite{rdp1}, which allows a simple additive form for the overall privacy leakage, and second, by designing an effective online algorithm that is shown to provide strong performance guarantees with respect to the offline optimum. Our contributions are as follows:

\begin{itemize}
\item We formulate an optimization problem whose objective is to minimize the time-averaged RDP leakage of devices over the entire training process after an arbitrary number of rounds $T$. 
The problem formulation includes a constraint to ensure model convergence to a stationary point of the global loss, along with individual transmit power constraints for each device.
We further derive a sufficient condition for convergence and reformulate the problem using this condition as a surrogate convergence constraint.

\item The reformulated problem involves both a long-term objective and a long-term constraint, making it difficult to solve due to the lack of knowledge about future channel conditions. Standard Lyapunov optimization techniques are not applicable, as the long-term constraint is unbounded over the feasible set.
Instead, we develop a novel online algorithm termed AdaScale, which decomposes the original problem into convex per-round optimization problems that can be solved efficiently using bisection search.

\item We further establish an upper bound for the dynamic regret of AdaScale, with respect to an offline optimum that assumes all future information is available. 
We demonstrate that the proposed method achieves $\mathcal{O}(T^{\max\{1-\beta, {\frac{1-\beta}{2}}\}})$ dynamic regret and  $\mathcal{O}(T^{\frac{\beta-1}{2}})$ constraint violation, where $\beta$ is a tunable parameter. Furthermore, when $ 1 < \beta < 2$, the regret bound diminishes to zero, and FL training converges to a stationary point, as $T \to \infty$. 

\item We conduct numerical experiments on canonical classification datasets under typical wireless network setting. Our results show that AdaScale is nearly optimal and outperforms state-of-the-art alternatives, effectively reducing both RDP and DP leakages under the same training convergence level.

\end{itemize}

\section{Related Work}\label{RELATEDWORKS}

Among existing works on DP in OTA FL, \cite{18,20, perRoundPrivacy} consider only per-round privacy leakage. They cannot provide proper trade-off between the \textit{overall} privacy leakage and training performance. 
More recent works evaluate the overall privacy leakage throughout the training process. Among them, \cite{22, 23} analyze privacy leakage only and do not design the receive scaling factors. In contrast, \cite{19, 21, 27, 28} focus on designing the receive scaling factors to enhance training performance while imposing constraints on the overall privacy leakage.

In \cite{19}, an optimal offline solution is obtained when the future channel conditions are known. Otherwise, estimation of the future channel is used to update the offline solution over the training rounds.
This work is extended in \cite{21}, to consider a reconfigurable intelligent surface (RIS), and in \cite{27}, to consider a multi-antenna server. 
These works provide essentially offline solutions, while our objective is \textit{online} adaptation to the time-varying channels over time.

The work in \cite{28} is the closest to ours. It extends the problem formulation of \cite{21} for active RIS and proposes an online solution.  
The standard Lyapunov optimization framework is employed to formulate per-round problems, which are solved using an alternating optimization heuristic. However, since the per-round problems are not solved within a bounded optimality gap, and the objective function is not bounded over the feasible set, the Lyapunov approach does not offer any performance guarantee \cite{neelybook}. In comparison, we propose a novel online solution in AdaScale that is proven to achieve diminishing regret with respect to the offline optimum, while guaranteeing that FL training converges to a stationary point.

Finally, in all aforementioned works, the overall privacy leakage is evaluated using the Advanced Composition Theorem for DP \cite{dwork2014algorithmic}, which is known to be loose and can lead to inefficient designs \cite{rdp1, 23}. In this work, we address this limitation through a tighter analysis based on the RDP.

\section{Preliminaries}
\label{chapt6-prelimiray}

\subsection{FL System}
We consider a wireless FL system comprising a central server and $M$ edge devices. Each device, indexed by $m$, contains a local training dataset $\mathcal{D}_m=\{(\mathbf{u}_{m,i}, o_{m,i}): 1\le i \le n_m \}$, where $\mathbf{u}_{m,i}$ is the $i$-th data feature vector, and $ o_{m,i}$ is its label. The local data of device $m$ follow distribution $p_m$.
The local loss function of device $m$ is defined as
\begin{align}\label{LocalLoss}
f_m(\mathbf{w})  \triangleq \mathbb{E}_{(\mathbf{u}_{m}, o_{m}) \sim p_m} c_ml\big(\mathbf{w}; (\mathbf{u}_{m}, o_{m})\big),
\end{align}
where $l(\cdot)$ represents a sample-wise loss function, $c_m \in \mathbb{R}$ is the device loss weight, and $\mathbf{w} \in \mathbb{R}^d$ contains the model parameters.
The edge devices aim to train a global model on the server cooperatively.
This requires minimizing a global loss function defined as
\begin{align}\label{GlobalLoss}
f(\mathbf{w}) = \frac{1}{M}\sum_{m=1}^M f_m(\mathbf{w}).
\end{align}
The ultimate goal is to determine the optimal model, 
$\mathbf{w}^\star$, that minimizes
$f(\mathbf{w})$ in a distributed manner.

In this study, we adopt the conventional Federated Stochastic Gradient Descent (FedSGD) technique \cite{mcmahan2017communication} for iterative model training in FL.
We consider OTA aggregation for uplink transmission from the devices to the server. The FedSGD algorithm with OTA aggregation is described in Section~\ref{chapt6-FedSGDwithOTA}.

\subsection{Differential Privacy}

Two widely adopted notions of Differential Privacy (DP) in the literature are $(\varepsilon, \delta)$-DP and $(\alpha, \varepsilon)$-RDP.

\begin{definition}[($\varepsilon$, $\delta$)-DP \cite{dwork2014algorithmic}]
The randomized mechanism $M:\mathcal{D}\rightarrow \mathcal{R} $ with domain $\mathcal{D}$ and range $\mathcal{R}$ satisfies \((\varepsilon, \delta)\)-DP if, for any two \textit{neighboring} datasets $ \mathcal{S} \in \mathcal{D}$ and $\mathcal{S}^\prime  \in \mathcal{D} $, i.e., \(\mathcal{S}^\prime \) is formed by adding or removing a single element from \(\mathcal{S}\), and for any output set $\mathcal{R}^\prime \subseteq \mathcal{R}$,
\begin{align}\label{dpDefinition}
\Pr [M(\mathcal{S}) \in \mathcal{R}^\prime ] \leq \mathrm{e}^{\varepsilon} \Pr [M(\mathcal{S}^\prime) \in \mathcal{R}^\prime] + \delta.
\end{align}
\end{definition}

\begin{definition}[($\alpha$, $\varepsilon$)-RDP \cite{rdp1}]
The randomized mechanism $\mathcal{M} \colon \mathcal{D} \to \mathcal{R}$ satisfies $(\alpha, \varepsilon)$-RDP for $\alpha \in \mathbb{R},\ \alpha > 1$ if for any neighboring datasets $ \mathcal{S} \in \mathcal{D}$ and $\mathcal{S}^\prime  \in \mathcal{D} $, it holds that
\begin{align}
D_\alpha\big(\mathcal{M}(\mathcal{S}) \,\|\, \mathcal{M}(\mathcal{S}^\prime )\big) \leq \varepsilon,
\end{align}
where $D_\alpha(p_1 \| p_2)$ denotes the Rényi divergence of order $\alpha$ between distributions $p_1(x)$ and $p_2(x)$:
\begin{align}
D_\alpha(p_1 \| p_2) = \frac{1}{\alpha - 1} \log \mathbb{E}_{x \sim p_2} \left[ \left( \frac{p_1(x)}{p_2(x)} \right)^\alpha \right].
\end{align}
\end{definition}
\begin{remark}[Conversion from RDP to DP \cite{rdp1}]\label{convert-remark}
   If a randomized mechanism $\mathcal{M}$ satisfies $(\alpha, \varepsilon_1)$-RDP for some $\alpha > 1$, then for any $\delta \in (0,1)$, it also satisfies $(\varepsilon_2, \delta)$-DP, where
\begin{align}
\varepsilon_2 = \varepsilon_1 + \log\Big(\frac{\alpha-1}{\alpha}\Big)-
\frac{\log\delta+ \log \alpha}{\alpha - 1}.
\end{align}
\end{remark}

Next, we present a method to compute the RDP leakage.

\begin{definition}[Sampled Gaussian Mechanism (SGM) \cite{rdp2-sampledGauss}]\label{SGMdEF} Let $u$ be a function mapping subsets of $\mathcal{D}$ to $\mathbb{R}^d$. We define the Sampled Gaussian Mechanism parameterized with the sampling rate $0< q\leq 1$ and the noise $\sigma>0$ as
\begin{align}\label{EQ-SGM}
\text{SG}_{q,\sigma}(\mathcal{D})\triangleq u(\mathcal{S})+\mathcal{N}(0,\sigma^2\mathbb{I}^d),
\end{align}
where $\mathcal{S} = \{x\colon x\in \mathcal{D} \textrm{ is sampled with probability } q\}$ is formed by sampling each element of $\mathcal{D}$ independently at random with probability $q$ without replacement, and $\mathcal{N}(0,\sigma^2\mathbb{I}^d)$ is spherical $d$-dimensional Gaussian noise with per-coordinate variance $\sigma^2$.
\end{definition}

\begin{definition}[$\ell_2$-sensitivity]\label{l2sens-def}
Let \( u\) be a function with domain \( \mathcal{D} \) and range \( \mathcal{R} \). The \(\ell_2\)-sensitivity of \( u \) is \( \Delta \) if for any two neighboring datasets $ \mathcal{S} \in \mathcal{D}$ and $\mathcal{S}^\prime  \in \mathcal{D} $, it holds that
\begin{align}
\big \| u(\mathcal{S})- u(\mathcal{S}^\prime) \big\|_2 \le \Delta.
\end{align}
    
\end{definition}

\begin{lemma}[RDP leakage of SGM]\label{RDP-lEAKAGE-LEMMA}
For any integer \( \alpha > 1 \), the SGM defined in Definition~\ref{SGMdEF}, with mapping \( u(\cdot) \) having \(\ell_2\)-sensitivity \(\Delta\), satisfies $\big(\alpha, \rho_{\alpha}(q, \sigma_{\text{eff}})\big)$-RDP, where $\sigma_{\text{eff}} \triangleq \frac{\sigma}{\Delta}$ is the effective noise multiplier, $\rho_{\alpha}(q, \sigma_{\text{eff}})   \triangleq \frac{A_\alpha(q, \sigma_{\text{eff}})}{\alpha-1}$, and
\begin{align}
\hspace{-0.75em}
A_\alpha(q, \sigma_{\text{eff}}) \triangleq \ln \Big[ \sum_{k=0}^{\alpha}  \binom{\alpha}{k} (1-q)^{\alpha-k} q^k \exp{\Big(\frac{k^2-k}{2\sigma_{\text{eff}}^2}\Big)} \Big].
\end{align}
\end{lemma}

\begin{proof}
The result can be directly derived from \cite{rdp2-sampledGauss}, and is omitted here for brevity.
\end{proof}

\section{System Model and Problem Formulation} \label{chapt6-systemmodel-problemdef}

In this section, we summarize the FedSGD \cite{mcmahan2017communication} algorithm with OTA uplink transmission, describe how we calculate its overall RDP leakage, and formulate our constrained RDP leakage minimization problem.

\subsection{FedSGD with OTA Aggregation}\label{chapt6-FedSGDwithOTA}

At each training round of FedSGD, the server updates the global model based on signals received from the devices. 
Specifically, round $t$ involves the following steps:
\begin{enumerate}
    \item \textbf{Model broadcast:} The server broadcasts the model parameter vector $\mathbf{w}_t$ to all devices. As commonly considered in the literature, we assume each device perfectly recovers the model.
    \item \textbf{Local gradient computation:} Each device $m$ forms a batch $\mathcal{B}_{m,t}$ according to Poisson sampling.\footnote{FedSGD is not specific to any sampling method, but we will see later that Poisson sampling is needed for tractable RDP analysis.} Specifically, each data point is sampled independently with probability $q_m= \frac{B_m}{n_m}$ from its local dataset $\mathcal{D}_m$, where $B_m$ is the expected batch size. 
    The devices then compute the average of the sample gradients over the batch:
    \begin{align}\label{averaging}
        {\mathbf{g}}_{m,t}= \frac{1}{B_m} \sum_{i\in \mathcal{B}_{m,t}}   {\mathbf{g}}_{m,t,i},
    \end{align}
    where $\mathbf{g}_{m,t,i}\triangleq c_m\nabla l\big(\mathbf{w}_t, (\mathbf{u}_{m,i}, o_{m,i})\big)  \in \mathbb{R}^d$.
    \item \textbf{OTA uplink transmission:} The devices transmit ${\mathbf{g}}_{m,t}$ to the server via OTA aggregation \cite{B1}. Specifically, all devices select a transmit weight $a_{m,t} \in \mathbb{C}$ and send $a_{m,t} {\mathbf{g}}_{m,t}$ to the server simultaneously using the same frequency resource over $d$ consecutive time slots. 
    \item \textbf{Receiver processing and model update at server:}
    Denote the channel coefficient of device $m$ by $h_{m,t} \in \mathbb{C}$. The received signal at the server is
    \begin{align}\label{receivedsignal-chapter6}
        \mathbf{r}_t = \sum_{m=1}^M h_{m,t} a_{m,t} {\mathbf{g}}_{m,t} +\mathbf{n}_t,
    \end{align}
    where $\mathbf{n}_t \sim \mathcal{CN} (0, \sigma_n^2 \mathbb{I}^d)$ is the receiver noise.
    The server scales the received signal and updates the model by applying one-step gradient descent as\footnote{For more efficient transmission, $\mathbf{g}_{m,t}$ can be sent via complex signals using both the real and imaginary parts of the signal. This will not change the fundamental process developed subsequently.}
    \begin{align}\label{ModelUpdate}
                \mathbf{w}_{t+1} = \mathbf{w}_t - \lambda \frac{\mathrm{Re}({\mathbf{ r}_t})}{\sqrt{\eta_t}},
    \end{align}
where $\eta_t \in \mathbb{R}^+$ is the receive scaling factor, $\lambda$ is the learning rate, and $\mathrm{Re}(\cdot)$ returns the real part of a complex variable.   
\end{enumerate}

As in~\cite{19, 21, 28}, we assume that the sample gradient norms are upper bounded by $G$, i.e., $\|\mathbf{g}_{m,t,i}\| \le G, \forall m, t, i$, and we set the device transmit weights proportional to the inverse of the uplink channels.
Thus, $a_{m,t} = \frac{\sqrt{\eta_t}}{M h_{m,t}}, \forall m, \forall t$. 
Then, we can rewrite the server processed received signal in \eqref{ModelUpdate} as
    \begin{align}
        \tilde{\mathbf{r}}_t & \triangleq \frac{\mathrm{Re}(\mathbf{r}_t)}{\sqrt{\eta_t}} =  \underbrace{\frac{1}{M}\sum_{m=1}^M {\mathbf{g}}_{m,t}}_\mathrm{\triangleq \: \mathbf{s}_t } +
         \underbrace{\frac{\mathrm{Re}(\mathbf{n}_t)}{\sqrt{\eta_t}}}_{\triangleq \: \tilde{\mathbf{n}}_t} \label{decmpose},
    \end{align}
    which contains two parts: i) the signal $\mathbf{s}_t$, and ii) the effective noise at the receiver $\tilde{\mathbf{n}}_t \sim \mathcal{N}(\mathbf{0}, \frac{\sigma_n^2}{2\eta_t} \mathbb{I}^d)$.

The average transmit power of device $m$ in round $t$ is 
\begin{align}
    P_{m,t} &= |a_{m,t}|^2 \frac{\mathbb{E}\|{\mathbf{g}}_{m,t} \|^2}{d} = \frac{\eta_t \mathbb{E}\|{\mathbf{g}}_{m,t} \|^2}{d M^2 |h_{m,t}|^2} \nonumber \\ &\overset{(a)}{\le} \frac{\eta_t G^2}{d M^2 |h_{m,t}|^2} \mathbb{E}[\frac{|\mathcal{B}_{m,t}|^2}{B_m^2}] \overset{(b)}{=} \frac{\eta_t G^2 k_m^2}{d M^2 |h_{m,t}|^2} ,
\end{align}
where (a) follows from the fact that based on \eqref{averaging}, ${\mathbf{g}}_{m,t}$ is the average of sample gradients with norms less than or equal to $G$, and thus, by the triangle inequality, we have $\|{\mathbf{g}}_{m,t}\| \le \frac{|\mathcal{B}_{m,t}|G}{B_m}$; and (b) follows from $k_m^2 \triangleq \mathbb{E}[\frac{|\mathcal{B}_{m,t}|^2}{B_m^2}] = 1+ \frac{(1-q_m)}{B_m}$ due to Poisson sampling.

\subsection{RDP Leakage Calculation}

Privacy leakage quantifies the information about the device local data samples that the server can extract from the post-processed received signal $\tilde{\mathbf{r}}_t$.\footnote{Note that the imaginary part of $\mathbf{r}_t$ is pure noise, containing no information.}

\subsubsection{Per-Round RDP Leakage}
By comparing \eqref{decmpose} with Definition~\ref{SGMdEF}, it is evident that for each device \( m \), the vector \( \tilde{\mathbf{r}}_t \) constitutes an SGM with respect to the local dataset \( \mathcal{D}_m \). 
Thus, RDP leakage for each device can be quantified using Lemma~\ref{RDP-lEAKAGE-LEMMA}, by identifying the effective noise multiplier associated with \( \tilde{\mathbf{r}}_t \) for each device.
By Definition~\ref{l2sens-def}, the $\ell_2$ sensitivity of $\mathbf{s}_t$ with respect to the batch of device $m$ is $\Delta_{m,t} = \frac{G}{B_m M}$, since the norm of each sample gradient is upper bounded by $G$, and the aggregation of sample gradients is divided by $M B_m$ according to \eqref{averaging} and \eqref{decmpose}.
Now, given $\Delta_{m,t}$, the effective noise multiplier for device $m$ is 
\begin{align}\label{EffectiveNoiseMultiplier}
    \sigma_{m,t} &= \frac{\frac{\sigma_n}{\sqrt{2 \eta_t}}}{\Delta_{m,t}}  = \frac{MB_m\sigma_n}{\sqrt{2 \eta_t} G}.
\end{align}
Based on Lemma~\ref{RDP-lEAKAGE-LEMMA}, with $\sigma_{m,t}$ in hand, for any order $\alpha$, the RDP leakage for device $m$ in round $t$ is $\rho_{\alpha}(q_m, \sigma_{m,t})$.

\subsubsection{Overall RDP Leakage}
The RDP leakage of a sequence of randomized mechanisms composed sequentially is given by the sum of the RDP leakages of the individual mechanisms \cite{rdp1}. Thus, the overall RDP leakage over $T$ rounds for device $m$ is $\sum_{t=0}^{T-1} \rho_{\alpha}(q_m, \sigma_{m,t})$.

\subsection{Problem Formulation}

We aim to minimize the overall RDP leakage after $T$ training rounds, via optimizing the receive scaling factors $\{\eta_t\}_{t=0}^{T-1}$, while ensuring a certain level of convergence of the global model: 
\begin{subequations}\label{FirstProblem}
    \begin{align}
    \min_{\{ \eta_{t} \}} \quad & \frac{1}{T}\sum_{t=0}^{T-1} \sum_{m=1}^M \rho_{\alpha}(q_m, \sigma_{m,t}) \label{FirstProblem-obj}\\
    \textrm{s.t.} \quad &
    \frac{1}{T} \sum_{t=0}^{T-1} \mathbb{E} \|\nabla f(\mathbf{w}_t) \|^2 \le \gamma,  \label{FirstProblem-constraint} \\
     \quad &
     \frac{\eta_t G^2 k_m^2}{dM^2 |h_{m,t}|^2} \le P_{\max}, \forall m, \forall t,\label{FirstProblem-powercosntraint} \\ \quad &
     \eta_t > 0, \forall t, \label{FirstProblem-eta-poistive}
    \end{align}
\end{subequations}
where the $\mathbb{E}[\cdot]$ is on the randomness of the batch sampling, the noise of sample gradients, and the receiver noise.
Constraint \eqref{FirstProblem-constraint} ensures that the system achieves \(\gamma\)-convergence to a stationary point of the global loss function \( f(\mathbf{w}) \), and constraint \eqref{FirstProblem-powercosntraint} limits the average power consumption of devices. 

\begin{remark}
We note that bounding $\frac{1}{T} \sum_{t=0}^{T-1} \mathbb{E}\|\nabla f(\mathbf{w}_t) \|^2$ in \eqref{FirstProblem-constraint} implies a bound on  $\min_{0\le t\le T-1} \allowbreak \mathbb{E}\|\nabla f(\mathbf{w}_t) \|^2$, which guarantees that at least a model among $\{ \mathbf{w}_t\}$ during the training process will be sufficiently close to a stationary point if \(\gamma\) is chosen to be small enough.   
\end{remark}

Solving \eqref{FirstProblem} presents significant challenges because constraint~\eqref{FirstProblem-constraint} involves the gradient of the global loss function, which is not an explicit function of the optimization variables. Additionally, the global loss function is typically not quantifiable, as it depends on the local data distributions $\{p_m\}$, which are unknown. 
Furthermore, the effective noise multipliers $\{\sigma_{m,t}\}$ in~\eqref{FirstProblem-obj} and the model sequence $\{\mathbf{w}_t\}$ in~\eqref{FirstProblem-constraint} depend on the channel conditions at each round, which are unknown prior to the start of the round. This necessitates the development of an online solution to address unknown future information.
To proceed, we first analyze the convergence of FedSGD with OTA aggregation and substitute \eqref{FirstProblem-constraint} with a more manageable surrogate constraint.

\section{Adaptive Receive Scalar Design}\label{chapt6-Design}

In this section, we first reformulate problem~\eqref{FirstProblem} through the training convergence analysis. We then present an online algorithm to adaptively design the receiver scaling factors $\{ \eta_t \}_{t=0}^{T-1}$ to address the trade-off between privacy and training convergence.

\subsection{Problem Reformulation via Training Convergence Analysis}

Convergence analysis for FedSGD under uniform batch sampling with ideal communication is provided in~\cite{Rethinking}. Here, we extend this analysis to account for Poisson sampling and OTA aggregation transmission. We then use the resulting convergence bound to reformulate problem~\eqref{FirstProblem}.

\subsubsection{Convergence Analysis} We consider the following assumptions on the loss function, which are common in the literature of distributed training and first-order optimization \cite{nips2017, understandingGradientClipping}:
\begin{enumerate}
       \item[\textbf{A1.}] \textbf{Smoothness:} 
       $\forall \mathbf{w}, \mathbf{w}^\prime \in \mathbb{R}^d$,
        \begin{align}
        \hspace{-2em}
            f_m(\mathbf{w}) & \le f_m(\mathbf{w}^\prime)+ \big\langle \nabla f_m(\mathbf{w}^\prime), \mathbf{w}-\mathbf{w}^\prime \big\rangle  + \frac{L}{2}\| \mathbf{w}-\mathbf{w}^\prime\|^2 \nonumber. 
        \end{align}
        \item[\textbf{A2.}] \textbf{Global minimum:} 
        $\exists \mathbf{w}^\star \in \mathbb{R}^d$ such that,
        \begin{align}
            f(\mathbf{w}^\star) = f^\star \le f(\mathbf{w}), \forall \mathbf{w} \in \mathbb{R}^d.
        \end{align}
        \item[\textbf{A3.}] \textbf{Unbiased sample gradients with bounded variance:} 
        $\exists A_1, A_2 \ge 0$, such that $\forall \mathbf{w}_t \in \mathbb{R}^d$,
        \begin{align}
        & \mathbf{g}_{m,t,i} = \nabla f_m(\mathbf{w}_t) +
        \mathbf{z}_{m,t,i},  \: \mathbb{E}\Big[\mathbf{z}_{m,t,i}| \mathbf{w}_t \Big] = 0,  \label{A3-Unbiasedness}\\ & 
         \mathbb{E}\Big[\|\mathbf{z}_{m,t,i}\|^2 | \mathbf{w}_t\Big] \le  A_1 \| \nabla f_m(\mathbf{w}_t)\|^2 + A_2. \label{BoundedVar}
        \end{align}
        \item[\textbf{A4.}] \textbf{Bounded similarity:} 
        $ \exists C_1, C_2 \ge 0$ such that $\forall \mathbf{w} \in \mathbb{R}^d$,
        \begin{align}
        \hspace{-3em}
          \frac{1}{M}\sum_{m=1}^M\| \nabla f_m(\mathbf{w})- \nabla f(\mathbf{w})\|^2 &\le C_1 \| \nabla f(\mathbf{w})\|^2  +C_2.
        \end{align}
\end{enumerate}

In the following, we provide our convergence bound for the global loss function under FedSGD with OTA aggregation.

\begin{theorem}[Training convergence]\label{chapt6-ConvergenceTheorem}
    Assume \textbf{A1}-\textbf{A4} hold, and the learning rate is set as $\lambda \le \frac{1}{4L (C_1+1)(A_1+1)}$. 
    After $T$ rounds of FedSGD described in Section~\ref{chapt6-FedSGDwithOTA}, we have
    \begin{align}\label{ConvBound}
        \frac{1}{T} \sum_{t=0}^{T-1}\mathbb{E}\| \nabla f(\mathbf{w}_t)\|^2 & \le \phi +  \frac{L \lambda}{2T} \sum_{t=0}^{T-1}\frac{ d \sigma_n^2}{\eta_t},
\end{align}
where $\phi$ is defined as
\begin{align}
\phi \triangleq \frac{2\big(f(\mathbf{w}_0)- f^\star\big)}{\lambda T}  + 2L \lambda   \Big(2C_2(A_1+1) + A_2 \Big).
    \end{align}
\end{theorem}
\begin{proof}
    See Appendix~\ref{appendix-6A}.
\end{proof}
Our bound differs slightly from those in prior works~\cite{19,21,28}, as it accounts for Poisson sampling and is derived under weaker assumptions. Specifically, unlike previous bounds that assume strong convexity or the Polyak-Łojasiewicz condition, our analysis does not require convexity of the loss function.

\subsubsection{Problem Reformulation}

To reformulate problem~\eqref{FirstProblem}, we apply a change of variable and define
\begin{align}
x_t \triangleq \frac{\eta_t}{h_{{\min},t}^2},
\end{align}
where $h_{{\min},t} \triangleq \min_{m} \frac{|h_{m,t}|}{k_m}$. We further define $x_{\max} \triangleq \frac{P_{\max}dM^2}{G^2}$.
Then, constraints ~\eqref{FirstProblem-powercosntraint} and ~\eqref{FirstProblem-eta-poistive} convert to 
\begin{align}
0 < x_t \le x_{\max}, \forall t.
\end{align} 
Moreover, the effective noise multiplier in \eqref{EffectiveNoiseMultiplier} can be written in terms of $x_t$ as
\begin{align}\label{EffectiveNoiseMultiplier-2}
    \sigma_{m,t} = \frac{MB_m\sigma_n}{\sqrt{2 x_t} G h_{{\min}, t}}.
\end{align}

To deal with constraint~\eqref{FirstProblem-constraint}, first we rewrite the bound in \eqref{ConvBound} in terms of $x_t$ as follows:
\begin{align}\label{NewConvBound}
        & \frac{1}{T}  \sum_{t=0}^{T-1} \mathbb{E}\| \nabla f(\mathbf{w}_t)\|^2  \le  \phi + \frac{L\lambda }{2T} \sum_{t=0}^{T-1}\frac{d \sigma_n^2}{h_{{\min},t}^2 x_{\max}} \nonumber  \\ & \qquad \qquad \qquad \qquad+ \frac{L \lambda }{2T} \sum_{t=0}^{T-1} \frac{d \sigma_n^2}{h_{{\min},t}^2}\Big(\frac{1}{x_t}-\frac{1}{x_{\max}} \Big).
\end{align}
We note that the second term of the upper bound in \eqref{NewConvBound} does not depend on the decision variables \(\{x_{t}\}\). For simplicity, we define these terms as 
\begin{align}
\phi^\prime \triangleq \phi + \frac{L \lambda }{2T} \sum_{t=0}^{T-1}\frac{d \sigma_n^2}{h_{{\min},t}^2 x_{\max}}.
\end{align}

We replace the left-hand side (LHS) of \eqref{FirstProblem-constraint} with its upper bound given in \eqref{NewConvBound}.
To ensure that \(\frac{1}{T} \sum_{t=0}^{T-1} \| \nabla f(\mathbf{w}_t)\|^2\) is bounded by \(\gamma\), it suffices to bound the right-hand side (RHS) of \eqref{NewConvBound} by the same amount. Since the first two terms of RHS of \eqref{NewConvBound} are constant, restricting it by \(\gamma\) implies a bound on the third term \(\frac{1}{T} \sum_{t=0}^{T-1}\frac{d \sigma_n^2}{h_{{\min},t}^2 } (\frac{1}{x_t}-\frac{1}{x_{\max}})\) by \(\nu\), where \(\nu = \frac{ 2(\gamma - \phi^{\prime}) }{\lambda L}\). Hence, we reformulate problem \eqref{FirstProblem} as
\begin{subequations}\label{OriginalProblem}
    \begin{align}
    \min_{\{ x_{t} \}} \quad & \frac{1}{T} \sum_{t=0}^{T-1} \sum_{m=1}^M \rho_{\alpha}(q_m, \sigma_{m,t}) \label{OriginalProblem-obj}\\
    \textrm{s.t.} \quad &
    \frac{1}{T} \sum_{t=0}^{T-1}\frac{d \sigma_n^2}{h_{{\min},t}^2 } \Big(\frac{1}{x_t}-\frac{1}{x_{\max}} \Big) \le \nu, \label{OriginalProblem-constraint} \\ \quad &
    0 < x_t \le x_{\max}, \forall t, \label{OriginalProblem-PowerConstraint} 
    \end{align}
\end{subequations}
where $\nu$ replaces $\gamma$ as the hyperparameter to tune the trade-off between training convergence and privacy.

The above problem is still difficult to handle due to the presence of the long-term objective and constraint, and the channel coefficients $\{h_{m,t}\}$ are unknown prior to the start of the $t$-th round. 
Next, we propose a novel algorithm to solve the problem in an online manner and provide bounds for both its constraint violation and its dynamic regret.

\subsection{Proposed Algorithm}

We start with a conventional virtual queue to keep track of the violation of constraint \eqref{OriginalProblem-constraint}, which is denoted by $Q_t \in \mathbb{R}$ with $Q_0  = 0$. In each round $t$, the server updates the virtual queue as
\begin{align}\label{VirtualQueueUpdate}
    Q_{t+1} = \max\Big\{Q_{t}+ \frac{ d \sigma_n^2}{h_{{\min},t}^2} \Big( \frac{1}{x_t}-\frac{1}{x_{\max}}\Big)- \nu, 0 \Big\}.
\end{align} 

If we directly apply standard Lyapunov optimization~\cite{neelybook} to solve problem~\eqref{OriginalProblem}, the decision variable at each round would be obtained by solving a per-round optimization problem with objective $V \sum_{m=1}^M \rho_{\alpha}(q_m, \sigma_{m,t}) + Q_t \frac{d \sigma_n^2}{h_{{\min},t}^2} \left( \frac{1}{x_t} - \frac{1}{x_{\max}} \right).$
Minimizing this objective is equivalent to minimizing an upper bound on the drift-plus-penalty, if the constraint function is \textit{bounded} within the feasible set \cite{neelybook}.
However, this boundedness assumption does not hold for problem~\eqref{OriginalProblem}, as the LHS of constraint~\eqref{OriginalProblem-constraint} can grow arbitrarily large when \( x_t \) approaches zero. In fact, it is easy to see that directly applying the standard Lyapunov method leads to infinite constraint violation.

This motivates us to modify the standard Lyapunov method by introducing an \textit{additional term} into the per-round objective. This modification prevents the solution from collapsing to zero and avoids unbounded constraint violations. As we will show in Section~\ref{chapt6-theoreticalanalysis}, the inclusion of this additional term makes the per-round optimization problem equivalent to minimizing an upper bound on the drift-plus-penalty, even in the presence of unbounded constraints. This enables us to establish performance guarantees.

Specifically, we consider a different form of the per-round optimization as follows.
In round $t$, the server solves an optimization problem to design its receiver scaling factor as
\begin{subequations}\label{PerRoundProposedOptimization}
    \begin{align}
    \min_{ x_t } \quad &   V \sum_{m=1}^M \rho_{\alpha}(q_m, \sigma_{m,t}) + Q_t \frac{ d \sigma_n^2}{h_{{\min},t}^2} \Big( \frac{1}{x_t}-\frac{1}{x_{\max}}\Big) \nonumber \\ & \qquad+ \frac{1}{2}\big(\frac{ d \sigma_n^2}{h_{{\min},t} ^2}\big)^2 \Big( \frac{1}{x_t}-\frac{1}{x_{\max}}\Big)^2 \label{PerRound-Objective}\\
    \textrm{s.t.} \quad &
    0 < x_t  \le x_{\text{max}}. \label{PerRound-Constraint}
    \end{align}
\end{subequations}
where $V \in \mathbb{R}^+$ is a predefined constant. Note that $\rho(q_m, \sigma_{m,t})$ depends on $x_{t}$ through \eqref{EffectiveNoiseMultiplier-2}.

Problem~\eqref{PerRoundProposedOptimization} is a single-variable optimization problem. The following proposition establishes that it is convex for integer values of $\alpha$.

\begin{proposition}\label{OptimalSolution}
For any integer $\alpha$, problem~\eqref{PerRoundProposedOptimization} is convex.
\end{proposition}
\begin{proof}
Since the constraint~\eqref{PerRound-Constraint} is linear in $x_t$, it suffices to show that the objective function in~\eqref{PerRound-Objective} is convex in $x_t$. 
The objective function has three terms. The first term is $\sum_{m=1}^M \rho_\alpha(q_m, \sigma_{m,t})$. For integer $\alpha$, $\rho_\alpha(q, \sigma)$ is defined in Lemma~\ref{RDP-lEAKAGE-LEMMA}. 
Plugging in $\sigma_{m,t}$ in terms of $x_t$ using~\eqref{EffectiveNoiseMultiplier-2}, we observe that $\rho(q_m, \sigma_{m,t})$ becomes a \texttt{logsumexp} function of $x_t$, which is a known convex function. Thus, the first term of the objective in~\eqref{PerRoundProposedOptimization} is a sum of convex functions across devices, and hence is convex.
The second term of the objective function involves $\frac{1}{x_t}$, which is convex over the feasible set as $x_t > 0$. The third term involves $\left( \frac{1}{x_t} - \frac{1}{x_{\max}} \right)^2$, which is a composition of two functions: $g_1(x) = x^2$ and $g_2(x) = \frac{1}{x} - \frac{1}{x_{\max}}$. Both $g_1(x)$ and $g_2(x)$ are convex, and $g_1(x)$ is increasing over the feasible set since $x_t \le x_{\max}$. Therefore, by the composition rule for convex functions, the third term is also convex over the feasible set. 
Hence, the overall objective function is convex, and so is the optimization problem.
\end{proof}

Based on Proposition~\ref{OptimalSolution}, for any integer $\alpha$, problem~\eqref{PerRoundProposedOptimization} can be solved by setting the derivative of the objective function to zero and identifying its root. If the root lies within the interval \( (0, x_{\max}] \), it corresponds to the optimal solution; otherwise, the optimal solution is given by \( x_{\max} \). However, due to the complexity of the objective function, finding a closed-form expression for the root is not feasible. 
Therefore, we employ the bisection algorithm, based on the derivative of the objective function,  to numerically compute the point at which the derivative of the objective function equals zero.
The detailed procedure is standard and is omitted to avoid redundancy.

We refer to our proposed algorithm, which adaptively designs the receive scaling factor by solving \eqref{PerRoundProposedOptimization} and updating the virtual queue based on \eqref{VirtualQueueUpdate}, as Adaptive receive Scaling (AdaScale). It is summarized in Algorithm~\ref{receivescalerdesign--algo}.

\begin{algorithm}[!t]
\caption{AdaScale at round $t$}
\label{receivescalerdesign--algo}
\textbf{Inputs}:  $\sigma_n$, $\{ q_m \}$, $\{ B_m \}$, $G$, $M$, $d$, $P_{\max}$. \\
\textbf{Output}: $\eta_t$ 
\begin{algorithmic}[1] 
\vspace{0.1 em}
    \STATE Server solves~\eqref{PerRoundProposedOptimization} using bisection search. 
    \vspace{0.15em}
    \STATE Server updates its virtual queue based on \eqref{VirtualQueueUpdate}.
    \vspace{0.15em}
    \STATE Server sets $\eta_t = x_t \min_{m} \frac{|h_{m,t}|^2}{k_m^2}.$
    \vspace{0.15em}
    \STATE Server transmits $\eta_t$ to the devices; devices use it to set their transmit weights.
\end{algorithmic}
\end{algorithm}

\begin{remark}
Throughout this work, we consider integer values of \( \alpha \). Nevertheless, since the RDP leakage is a monotonically increasing function of \( \alpha \), upper and lower bounds on the leakage for a non-integer \( \alpha \) can be obtained by evaluating the RDP expression at the closest integers. 
\end{remark}

\subsection{Computational Complexity}
Solving \eqref{PerRoundProposedOptimization} using the bisection algorithm requires evaluating the derivative of \eqref{PerRound-Objective}, which has a constant computational cost of $\mathcal{O}(1)$. To reach a solution within a distance of $\tau$ from the optimum, the algorithm requires at most $\log_2(\frac{x_{\max}}{\tau})$ iterations. Therefore, in each training round, the overall computational complexity for obtaining a solution within $\tau$-vicinity of the optimum is $\mathcal{O}\big(\log_2(\frac{x_{\max}}{\tau})\big)$.

Despite the low computational complexity of this algorithm, we next show that it has strong performance guarantees, in terms of constraint violation and dynamic regret.

\section{Theoretical Performance Analysis}\label{chapt6-theoreticalanalysis}

We analyze the performance of AdaScale in this section. We note that even though our analysis uses the familiar notion of drift, it is substantially different from the conventional Lyapunov stability analysis, and it leads to novel constraint violation and dynamic regret bounds. To begin, let \( \hat{x}_{t} \) denote the optimization variable at round \( t \) obtained using AdaScale, and let \( \hat{\sigma}_{m,t} \) denote the corresponding effective noise multiplier, obtained by substituting \( \hat{x}_{t} \) for \( x_t \) in \eqref{EffectiveNoiseMultiplier-2}.

\subsection{Upper Bound on $R$-Slot Drift}

For any positive integer $R \le T$, we define the $R$-slot drift of the virtual queue as
\begin{align}\label{R-slotDrift}
    \Delta_R(t) & \triangleq \frac{1}{2}Q_{t+R}^2 - \frac{1}{2}Q_{t}^2.
\end{align}
Using \eqref{R-slotDrift} and noting that the initial value of the queue is set to zero, we can rewrite the $R$-slot drift at time $t=0$ as $
\Delta_R(0)  = \frac{1}{2} Q_{R}^2,$
which implies
\begin{align}\label{QueueUpperbound}
    Q_{R} & = \sqrt{2 \Delta_R(0)}.
\end{align}

We start with an upper bound on the one-slot drift in the lemma below.
\begin{lemma}[One-slot drift bound]
    The one-slot drift for AdaScale is upper bounded by
\begin{align}\label{DrfitPlusPenalty}
         \Delta_1(t) & \le   Q_t \frac{ d \sigma_n^2}{h_{{\min},t}^2} \Big( \frac{1}{\hat{x}_t}-\frac{1}{x_{\max}}\Big) \nonumber \\ & \quad + \frac{1}{2}\big(\frac{ d \sigma_n^2}{h_{{\min},t} ^2}\big)^2 \Big( \frac{1}{\hat{x}_t}-\frac{1}{x_{\max}}\Big)^2 + \frac{1}{2} \nu^2.
    \end{align} 
\end{lemma}

\begin{proof}
Based on the queue update equation in \eqref{VirtualQueueUpdate}, we have $Q_{t+1} \le \big|Q_{t}+ \frac{ d \sigma_n^2}{h_{{\min},t}^2} \big( \frac{1}{\hat{x}_t}-\frac{1}{x_{\max}}\big)- \nu \big|$. Squaring both sides of this inequality, we obtain
    \begin{align} \label{one}     
        Q_{t+1}^2 & \le Q_{t}^2 + 2Q_t \Big( \frac{ d \sigma_n^2}{h_{{\min},t}^2} \big( \frac{1}{\hat{x}_t}-\frac{1}{x_{\max}}\big) -\nu \Big) \nonumber \\ & \qquad \qquad+  \Big(\frac{ d \sigma_n^2}{h_{{\min},t}^2} \Big( \frac{1}{\hat{x}_t}-\frac{1}{x_{\max}}\big)-\nu \Big)^2.
    \end{align}
    Rearranging the terms in \eqref{one}, we have
    \begin{align}
        \Delta_1(t) &\le \frac{ d \sigma_n^2}{h_{{\min},t}^2} \big( \frac{1}{\hat{x}_t}-\frac{1}{x_{\max}}\big) \big(Q_t-\nu \big)  \nonumber \\ & \quad+ \frac{1}{2}\big(\frac{ d \sigma_n^2}{h_{{\min},t} ^2}\big)^2 \Big( \frac{1}{\hat{x}_t}-\frac{1}{x_{\max}}\Big)^2+ \frac{1}{2} \nu^2- \nu Q_t,
    \end{align}
    where further upper bounding by disregarding the negative terms and noting $0 < \hat{x}_t \le x_{\max}$, leads to the upper bound given in \eqref{DrfitPlusPenalty}.
\end{proof}

We sum both sides of \eqref{DrfitPlusPenalty} over $t$ from $0$ to $R-1$ to obtain
\begin{align}\label{R-SlotDriftUpperBound}
      \Delta_R(0) 
     & \le   \sum_{t=0}^{R-1} Q_t \frac{ d \sigma_n^2}{h_{{\min},t}^2} \Big( \frac{1}{\hat{x}_t}-\frac{1}{x_{\max}}\Big) \nonumber \\ & \quad + \frac{1}{2} \sum_{t=0}^{R-1} \big(\frac{ d \sigma_n^2}{h_{{\min},t} ^2}\big)^2 \Big( \frac{1}{\hat{x}_t}-\frac{1}{\hat{x}_{\max}}\Big)^2   + \frac{R \nu^2}{2}.
\end{align}

\subsection{Upper Bound on Virtual Queue}

\begin{lemma}[Virtual queue upper bound]\label{QueueUpperBoundLemma}
Under AdaScale, the virtual queue is upper bounded by
    \begin{align}\label{QueueUpperbound-3}
    Q_t\le Q_T^{\max}, 0\le t\le T,
\end{align}
where 
\begin{align}
Q_T^{\max} & \triangleq \left( 2V \sum_{t=0}^{T-1}\sum_{m=1}^M \rho_\alpha(q_m, \sigma_{m,t}^{\min})  + T \nu^2 \right)^{\frac{1}{2}}, \\  {\sigma}^{\min}_{m,t} & \triangleq \frac{\sigma_n M B_m}{G h_{{\min},t} \sqrt{2{x}}_{\max}}.
\end{align}
    
\end{lemma}

\begin{proof}
From \eqref{R-SlotDriftUpperBound}, we have
\begin{align}\label{SumDriftPenalty}
    & V \sum_{t=0}^{R-1} \sum_{m=1}^M   \rho_{\alpha}(q_m, \hat{\sigma}_{m,t}) + \Delta_R(0) 
     \le \nonumber \\ & V \sum_{t=0}^{R-1} \sum_{m=1}^M   \rho_{\alpha}(q_m, \hat{\sigma}_{m,t})  +  \frac{1}{2} \sum_{t=0}^{R-1} \big(\frac{ d \sigma_n^2}{h_{{\min},t} ^2}\big)^2 \Big( \frac{1}{\hat{x}_t}-\frac{1}{x_{\max}}\Big)^2 \nonumber \\
    & \qquad \qquad \qquad+    \sum_{t=0}^{R-1} Q_t \frac{ d \sigma_n^2}{h_{{\min},t}^2} \Big( \frac{1}{\hat{x}_t}-\frac{1}{x_{\max}}\Big)  + \frac{R \nu^2}{2}.
\end{align}
Since AdaScale solves \eqref{PerRoundProposedOptimization} optimally, and the RHS of \eqref{SumDriftPenalty} is the summation of the objective function of \eqref{PerRoundProposedOptimization} (up to a constant) over rounds, AdaScale achieves the minimum value of the RHS of \eqref{SumDriftPenalty}. In particular, considering $x_{t}=x_{\max}, \forall t$, as a feasible solution to \eqref{PerRoundProposedOptimization} and its resultant RDP leakage $\sigma_{m,t}^{\min}, \forall t$, we obtain
\begin{align}\label{ZeroUpperbound}
&V \sum_{t=0}^{R-1} \sum_{m=1}^M   \rho_{\alpha}(q_m, \hat{\sigma}_{m,t})  + \Delta_R(0) \nonumber \\ & \le V \sum_{t=0}^{R-1} \sum_{m=1}^M  \rho_{\alpha}(q_m, \sigma^{\min}_{m,t})  + \frac{R \nu^2}{2}.
\end{align}
This implies
\begin{align}\label{rslotdriftupper}
    \Delta_R(0) \le V \sum_{t=0}^{R-1} \!\sum_{m=1}^M  \! \rho_{\alpha}(q_m, \sigma^{\min}_{m,t})+\!\frac{R \nu^2}{2} , 1\le R \le T.
\end{align}
Using \eqref{QueueUpperbound} together with \eqref{rslotdriftupper}, we can provide an upper bound on the queue length as
\begin{align}
    Q_R & \le \left(2 V \sum_{t=0}^{R-1} \!\sum_{m=1}^M  \! \rho_{\alpha}(q_m, \sigma^{\min}_{m,t})+\! R \nu^2 \right)^{\frac{1}{2}} \label{increasing-upp} \\  &  \overset{(a)}{\le} Q_T^{\max},  \quad 1\le R \le T,
\end{align}
where (a) follows from the fact that the RHS of \eqref{increasing-upp} is an increasing function of $R$.
\end{proof}

\subsection{Constraint Violation Bound}
The following theorem provides an upper bound on the amount of violation with respect to the constraint \eqref{OriginalProblem-constraint}.

\begin{theorem}[Constraint violation bound]\label{ConstrainViolationProp}
    Under AdaScale, the constraint violation of problem \eqref{OriginalProblem} is upper bounded as
    \begin{align}
           \frac{1}{T}\sum_{t=0}^{T-1} \frac{d \sigma_n^2}{h_{{\min},t}^2}  \big( \frac{1}{\hat{x}_t}-\frac{1}{x_{\max}}\big)- \nu  \le \frac{Q_T^{\max}}{T}.
\end{align}
\end{theorem}
\begin{proof}
We have
\begin{align}
\hspace{-1em}
      \frac{1}{T}\sum_{t=0}^{T-1} \frac{d \sigma_n^2}{h_{{\min},t}^2}  \big( \frac{1}{\hat{x}_t}-\frac{1}{x_{\max}}\big)-  \nu &  \overset{\mathrm{(a)}}{\le} \frac{1}{T} \sum_{t=0}^{T-1} \big( Q_{t+1} - Q_t \big)  \\
    &  \overset{\mathrm{(b)}}{\le} \frac{Q_T^{\max}}{T},
\end{align}
where (a) follows from $ Q_{t} + \frac{d \sigma_n^2}{h_{{\min},t}^2}  \big( \frac{1}{\hat{x}_t}-\frac{1}{x_{\max}}\big) -\nu \le  Q_{t+1} $ based on \eqref{VirtualQueueUpdate}, and (b) follows from Lemma \ref{QueueUpperBoundLemma}.
\end{proof}

\subsection{Dynamic Regret Bound}

Let $\{x^\star_t \}$ denote the offline optimal solution to \eqref{OriginalProblem} when all future information is available and $\sigma^{\star}_{m,t} \triangleq \frac{\sigma_n M B_m}{G h_{{\min},t} \sqrt{2x^{\star}_t}}$. We aim to derive an upper bound on the dynamic regret, which is the difference in the time-averaged RDP leakage achieved under AdaScale and that of $\{ x^\star_t\}$.

\begin{theorem}[Dynamic regret bound]\label{PerformanceGapProp} The dynamic regret of AdaScale is upper bounded as 
\begin{align}\label{Prop4-eq}
            \frac{1}{T}\sum_{t=0}^{T-1}\sum_{m=1}^M \Big( \rho_{\alpha}(q_m, \hat{\sigma}_{m,t}) -& \rho_{\alpha}(q_m, {\sigma}^\star_{m,t})  \Big)  \nonumber \\ & \le \frac{Q_T^{\max} \nu}{V}  +  \frac{T\nu^2}{2V} + \frac{\nu^2}{2V}.
\end{align}
\end{theorem}
\begin{proof}
We use a similar argument as in the proof of Lemma~\ref{QueueUpperBoundLemma}. Using \eqref{SumDriftPenalty} with $R= T$, and considering $x_t = x^\star_t, \forall t$, as a feasible solution to \eqref{PerRoundProposedOptimization}, we obtain
\begin{align}\label{DriftPenaltyUpperbound-PeriterationOptimal}
  & V \sum_{t=0}^{T-1}  \sum_{m=1}^M  \rho_{\alpha}(q_m, \hat{\sigma}_{m,t}) + \Delta_T(0)   \nonumber \\ & \le  V \sum_{t=0}^{T-1}  \sum_{m=1}^M  \rho_{\alpha}(q_m, {\sigma}^\star_{m,t})   +  \frac{1}{2} \sum_{t=0}^{T-1} \big(\frac{ d \sigma_n^2}{h_{{\min},t} ^2}\big)^2 \Big( \frac{1}{x^\star_t}-\frac{1}{x_{\max}}\Big)^2 \nonumber \\ & \qquad \qquad +    \sum_{t=0}^{T-1} Q_t \frac{ d \sigma_n^2}{h_{{\min},t}^2} \Big( \frac{1}{x^\star_t}-\frac{1}{x_{\max}}\Big)  + \frac{T \nu^2}{2}.
\end{align}
We now provide an upper bound on the RHS of \eqref{DriftPenaltyUpperbound-PeriterationOptimal}.
The second term in the RHS of \eqref{DriftPenaltyUpperbound-PeriterationOptimal} can be upper bounded as
\begin{align}
&\frac{1}{2}\sum_{t=0}^{T-1} \big(\frac{ d \sigma_n^2}{h_{{\min},t} ^2}\big)^2  \big( \frac{1}{x^{\star}_t}-\frac{1}{x_{\max}}\big)^2 \nonumber \\ & \qquad \overset{(a)}{\le}  \frac{1}{2}\Big( \sum_{t=0}^{T-1} \frac{ d \sigma_n^2}{h_{{\min},t} ^2} \big( \frac{1}{x^{\star}_t}-\frac{1}{x_{\max}}\big) \Big)^2  \overset{(b)}{\le}  \frac{T^2\nu^2 }{2}, \label{regretboundproof-5}
\end{align}
where (a) follows from the fact $\| \mathbf{y}\|_2 \le \| \mathbf{y}\|_1$, if all entries of $\mathbf{y} \in \mathbb{R}^T$ are positive, and (b) is due to the fact that $\{ x^\star_{t} \}$ meets the constraint~\eqref{OriginalProblem-constraint}. 
Additionally, the third term on the RHS of \eqref{DriftPenaltyUpperbound-PeriterationOptimal} can be further upper bounded as
\begin{align}
   \sum_{t=0}^{T-1} \! Q_t \frac{ d \sigma_n^2}{h_{{\min},t} ^2} \big( \frac{1}{x^{\star}_t}-\frac{1}{x_{\max}}\big) &  \overset{(a)}{\le} \! Q_T^{\max} \sum_{t=0}^{T-1} \! \frac{ d \sigma_n^2}{h_{{\min},t} ^2} \big( \frac{1}{x^{\star}_t}-\frac{1}{x_{\max}}\big) \nonumber  \\ & \overset{(b)}{\le}   TQ_T^{\max} \nu, \label{regretboundproof-4}
\end{align}
where (a) follows the result in Lemma~\ref{QueueUpperBoundLemma}, and (b) is due to the fact that $\{ x^\star_{t} \}$ meets the constraint~\eqref{OriginalProblem-constraint}. 

Applying the upper bounds in \eqref{regretboundproof-5} and \eqref{regretboundproof-4} on \eqref{DriftPenaltyUpperbound-PeriterationOptimal},  and dividing both sides by $TV$ and noting that $\Delta_T(0) \geq 0$ completes the proof. 
\end{proof}

\subsection{Discussion on Bounds}

In the following, we first present Corollary~\ref{firstCorollary} to simplify the bounds in Theorems \ref{ConstrainViolationProp} and \ref{PerformanceGapProp} and elucidate their scaling w.r.t. $T$. We then draw connection with the convergence of FL training in Corollary~\ref{ConvergenceCorollary}.

\begin{corollary}\label{firstCorollary}
Assume the minimum channel norm is bounded above, i.e., $\min_{m} |h_{m,t}|\le h_{\text{ub}}, \forall t$. Setting $V \propto T^\beta$ for any $\beta \in \mathbb{R}$, the constraint violation bound in Theorem \ref{ConstrainViolationProp} and the dynamic regret bound in Theorem \ref{PerformanceGapProp} reduce to the following:
\begin{subequations}\label{corollary}
\begin{align}
       & \lim_{T \to \infty }\frac{1}{T}\sum_{t=0}^{T-1} \frac{d \sigma_n^2}{ h_{{\min},t}^2}  \Big( \frac{1}{\hat{x}_t}-\frac{1}{x_{\max}}\Big)- \nu   \le \mathcal{O}\big(T^{\frac{\beta-1}{2}}\big). \label{Corol-1} \\ 
        & \lim_{T \to \infty }
        \frac{1}{T}\sum_{t=0}^{T-1}\sum_{m=1}^M \Big( \rho(q_m, \hat{\sigma}_{m,t}) - \rho(q_m, {\sigma}^\star_{m,t})  \Big) \nonumber \\ & \qquad \qquad \qquad \le \mathcal{O}\big(T^{\max\big\{1-\beta, \frac{1-\beta}{2}\big\}}\big). \label{Corol-2}
\end{align}
\end{subequations}
\end{corollary}
\begin{proof}
Setting \( V \propto T^\beta \) in the bounds of Theorems~\ref{ConstrainViolationProp} and~\ref{PerformanceGapProp}, and upper bounding \( Q_T^{\max} \) using $h_{\min,t} \overset{(a)}{\le} \min_m |h_{m,t}| \le h_{\text{ub}}$, we obtain the results in \eqref{Corol-1} and \eqref{Corol-2}. 
\end{proof}

In Corollary~\ref{firstCorollary}, the parameter $\beta$ balances the trade-off between utility and privacy. Specifically, $\beta > 1$ yields a diminishing bound for the regret, while $\beta < 1$ results in a diminishing bound for the constraint violation. Although these two regions of $\beta$ do not overlap, the following corollary establishes that, when the minimum channel norm is bounded both below and above, and $1 < \beta < 2$, AdaScale achieves diminishing dynamic regret and ensures convergence to a stationary point 
of the global loss function.

\begin{corollary}\label{ConvergenceCorollary}
Assume the minimum channel norm is bounded both below and above, i.e., $ h_{\text{lb}} \le \min_m |h_{m,t}| \le h_{\text{ub}}, \forall t$. Setting $V \propto T^\beta$, with  $1<\beta <2$, yields a diminishing regret bound, when $T \to \infty$. Moreover, by setting $\lambda \propto \frac{1}{\sqrt{T}}$, $\frac{1}{T} \sum_{t=0}^{T-1}\mathbb{E}\| \nabla f(\mathbf{w}_t)\|^2$ converges to zero when  $T \to \infty$.
\end{corollary}

\begin{proof}
Utilizing the result in \eqref{Corol-2}, it is clear that setting $\beta >1$ results in a diminishing time-averaged regret bound when $T \to \infty$.
Further substituting $\lambda \propto \frac{1}{\sqrt{T}}$ into \eqref{NewConvBound}, we observe that the first two terms on the RHS of \eqref{NewConvBound} are $\mathcal{O}(\frac{1}{\sqrt{T}})$, while the third term is also $O(\frac{1}{\sqrt{T}})$ since $ \frac{h_\text{lb}}{\sqrt{2}}\le \frac{\min_m |h_{m,t}|}{\max_m k_m} \le h_{\min, t}$ as $k_m \le \sqrt{2}, \forall m$. Additionally, substituting $\lambda$ into the fourth term, and using the result in \eqref{Corol-1}, we conclude that the fourth term is $O(T^{\frac{\beta-2}{2}})$. Since $\beta < 2$, all the terms converge to zero as  $T \to \infty$, which completes the proof.
\end{proof}

\section{Numerical Experiments}
\label{chapt6-experiments}

We evaluate the effectiveness of AdaScale in reducing privacy leakage during OTA FL training for classification tasks on the MNIST~\cite{MNIST} and CIFAR-10~\cite{CIFAR10} datasets.

We consider $M=10$ devices, and set the maximum power limit to $P_{\max}$ = 23 dBm.
Assuming a bandwidth of $100$ kHz, we set the noise power to $\sigma_n^2 = -90$ dBm, which accounts for both thermal noise and additional interference at the receiver.
The distance of device $m$ from the server is randomly generated, i.e., $d_{m} \sim \text{Uniform}[r_{\min},r_{\max}]$ with $r_{\min} = 10$ m,  $r_{\max} = 200$ m. The path loss follows the COST Hata model, i.e., $\text{PL}_m \text{[dB]} = 33.44 + 35.22 \log_{10}\big(d_m\big)$ \cite{ COSTHATA-2, COSTHATA-1}. The channel between device \( m \) and the server in round \( t \) is generated as $ h_{m, t} \sim \mathcal{CN}(0, \frac{1}{\text{PL}_m})$, which is i.i.d. across rounds.\footnote{The i.i.d. assumption is more realistic than a correlated channel model in FL, since the channel coherence time is typically much less than $200$ milliseconds even for a fixed device in a wireless environment \cite{rappaport2001wireless, temporalcorrelation}, while a training round of FL typically has duration on the order of seconds and minutes or more.}
We use the following benchmarks for comparison:
\begin{itemize}

\item \textbf{Optimal:} Assuming full knowledge of future information, problem~\eqref{OriginalProblem} becomes a convex problem that can be solved optimally. The resulting solution serves as a lower bound on the achievable privacy leakage.

\item \textbf{EqualAlloc:} This benchmark uniformly allocates \( \nu \) across all rounds to satisfy the constraint in \eqref{OriginalProblem-constraint}. Thus, it sets $x_t =  \frac{x_{\max}}{1+\frac{x_{\max}\nu h_{{\min},t}^2}{d \sigma_n^2}}, \forall t. $
    
\item \textbf{EstimFuture:} This method finds the MMSE estimation of the squared norm of future channels. Using these estimations, the convex problem~\eqref{OriginalProblem} is solved at each round to design \(x_t\), given the remaining constraint budget.

\item \textbf{Method in \cite{19}:} This approach aims to enhance training convergence while bounding the DP leakage. To address the unknown future channels, it employs MMSE estimation of the squared channel norm.

\item \textbf{Method in \cite{28}:} 
This approach has the same aim as that of \cite{19}. It is an online algorithm based on standard Lyapunov optimization.

\end{itemize}

Since, in practice, the upper bound on the norm of sample gradients $G$ is unknown, we follow the convention in the DP literature~\cite{19,21,28} and apply a clipping operation to the sample gradients using a predefined threshold $C$. Specifically, each sample gradient is replaced by its clipped version as 
$\mathbf{g}_{m,t,i} \leftarrow \mathbf{g}_{m,t,i} \min\left(1, \frac{C}{\|\mathbf{g}_{m,t,i}\|} \right)$. Correspondingly, in our solution formulation, $G$ is replaced by $C$.

Since problem~\eqref{OriginalProblem} minimizes the overall RDP leakage subject to a long-term convergence constraint bounded by \( \nu \), we consider different values of \( \nu \) as a measure of convergence and compare the resulting RDP and DP leakages across different methods for each $\nu$. Specifically, to evaluate the RDP leakage for a given order \( \alpha \), we compute \( \rho_m = \sum_{t=0}^{T-1} \rho_\alpha(q_m, \sigma_{m,t}) \) for each device, where the \( m \)-th device satisfies \( (\alpha, \rho_m) \)-RDP. The average RDP leakage across all devices is then reported as \( \frac{1}{M} \sum_{m=1}^M \rho_m \).
To evaluate DP leakage, we fix \( \delta = 10^{-5} \), and compute \( \varepsilon_m \) for each device, where the \( m \)-th device satisfies \((\varepsilon_m, \delta)\)-DP.\footnote{The value of \( \delta \) used in our experiments is commonly adopted in the literature for the datasets considered~\cite{tramerdifferentially}. In principle, \( \delta \) should be chosen to be on the order of \( \frac{1}{2n} \), where \( n \) denotes the dataset size.} We then report the average DP leakage across all devices as \( \frac{1}{M} \sum_{m=1}^M \varepsilon_m \).
The \texttt{Opacus} library~\cite{opacus} is used to compute \( \rho_m \) and \( \varepsilon_m \).

\noindent \textit{\textbf{Hyperparameter tuning:}} For each value of $\nu$, we tune the hyperparameters of each method so that the LHS of constraint~\eqref{OriginalProblem-constraint} matches \( \nu \). This approach allows for fair comparison of different methods under the \textit{same} learning performance.
Specifically, for AdaScale, we tune the parameter \( V \); for the method in~\cite{19}, we tune the privacy budget $\varepsilon_{\text{budget}}$; and for the method in~\cite{28}, we tune both the privacy budget $\varepsilon_{\text{budget}}$ and the objective multiplier \( V \) used in the Lyapunov framework. We set \( \alpha = 3 \) for both ``EstimFuture'' and AdaScale in all experimental settings.\footnote{We observed that, for a fixed \( \nu \), changing \( \alpha \) has a negligible impact on the privacy leakage as measured by DP.
}

\subsection{MNIST Dataset with I.I.D. Data Distribution }\label{chapt6-mnistsection}
In MNIST, each data sample is a labeled grey-scaled handwritten digit image of size $\mathbb{R}^{28} \times \mathbb{R}^{28} $ pixels, with a label indicating its class. There are $60,000$ training and $10,000$ test samples. We consider training a CNN whose architecture is detailed in \cite{papernot2021tempered, tramerdifferentially}, with $d = 26,010$ parameters.

An equal number of data samples from different classes are uniformly and randomly distributed among the devices. The batch size for each device is set to $60$. 
We set the number of training epochs to $5$ and thus the number of rounds is $T= 500$. The learning rate is constant throughout the training and set to $\lambda= 0.5$. The SGD optimizer with a weight decay of $10^{-4}$ is utilized for training. The clipping threshold is set to $C = 1.0$.
We consider several values of \( \nu \) ranging from \( 0.01 \) to \( 0.16 \), which correspond to test accuracies between \( 95\% \) and \( 90\% \).

\begin{figure}[t]
    \centering
\includegraphics[width=0.98\linewidth]{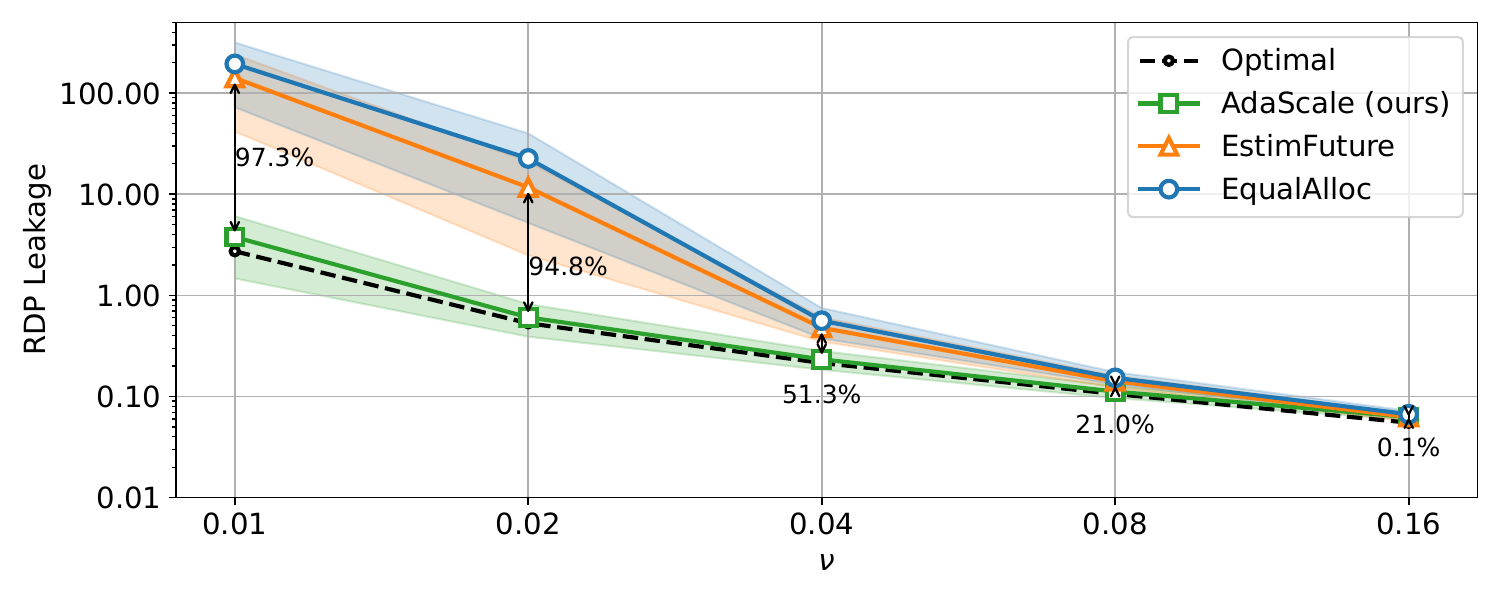}
\includegraphics[width=0.98\linewidth]{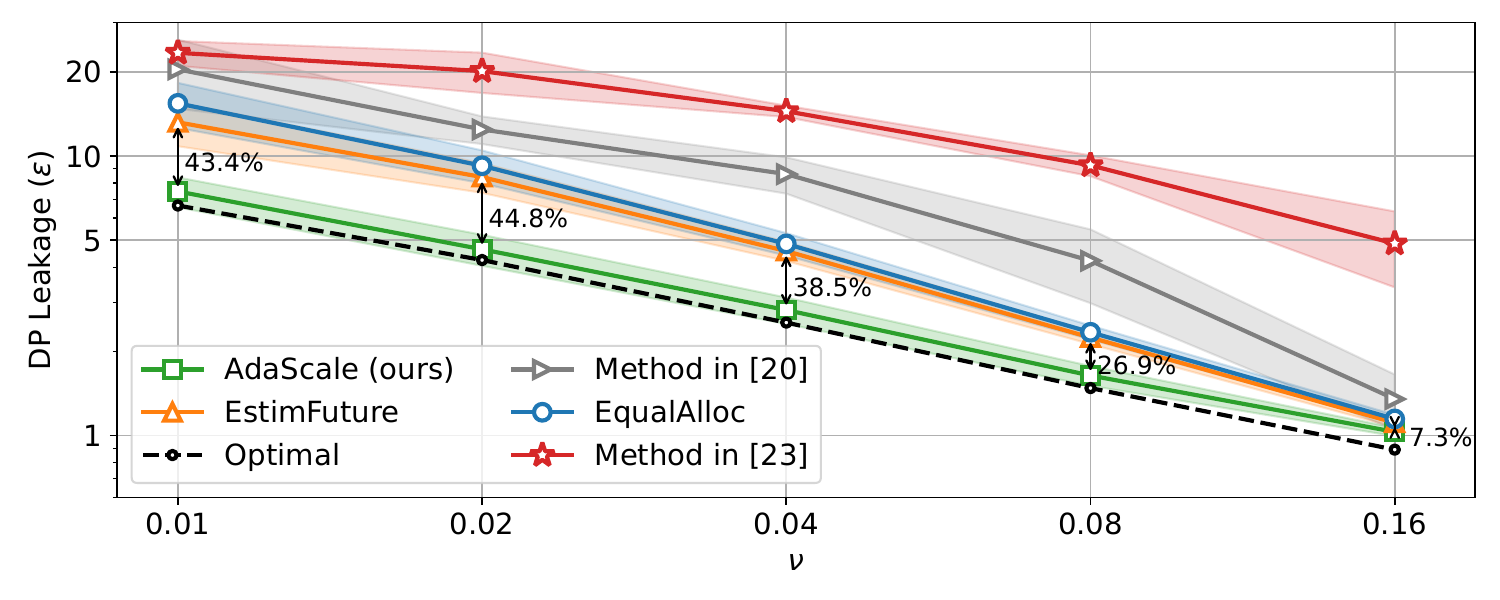}
\caption{RDP and DP leakage vs. $\nu$ for MNIST. Range of $\nu$ corresponds to test accuracies between $90\%$ and $95\%$.}\label{chap6-mnistfig}
\end{figure}

Fig.~\ref{chap6-mnistfig} illustrates the average overall RDP and DP leakages across devices plotted against \( \nu \). The results are averaged over three realizations, and the shaded regions around each curve represent the $95\%$ confidence intervals. Note that when evaluating RDP leakage, we consider only the first two benchmarks, “EqualAlloc” and “EstimFuture,” as the other two methods (from~\cite{19} and~\cite{28}) do not account for RDP leakage in their formulations and exhibit significantly higher leakage, making them incomparable. In contrast, for the evaluation of DP leakage, all four benchmarks are included in the comparison.

Fig.~\ref{chap6-mnistfig} shows that AdaScale reduces the RDP leakage compared with benchmarks across different values of \( \nu \). Furthermore, AdaScale performs close to the offline Optimal benchmark. As \( \nu \) increases, the gap between AdaScale and the benchmarks narrows, since the benchmarks also approach near-optimal performance. However, it is important to note that the more desirable regime corresponds to smaller values of \( \nu \), which correspond to higher learning accuracy, where AdaScale's advantage becomes more pronounced.

Figure~\ref{chap6-mnistfig} further shows that all methods incur higher DP leakage as \( \nu \) decreases, which aligns with the results on RDP.
We observe that although AdaScale is primarily designed with RDP as objective, it can effectively improve privacy in terms of the DP metric as well, outperforming state-of-the-art benchmarks and performs closely to the optimal offline solution. Again, this improvement becomes clearer for smaller values of $\nu$, which correspond to higher learning accuracies.

\subsection{CIFAR-10 Dataset with Non-I.I.D. Data Distribution}
In CIFAR-10, each data sample consists of a colored image of size $ \mathbb{R}^{3}\times \mathbb{R}^{32} \times \mathbb{R}^{32}$ and a label indicating the class of the image. There are $50,000$ training and $10,000$ test samples. We train the CNN described \cite{ResNet} with approximately $500,000$ parameters using the cross-entropy loss. 

The training data is distributed across devices in a non-i.i.d. manner, with each device containing $5000$ samples only from two classes. The batch size is set to $400$, and the training is conducted over $60$ epochs, resulting in $T=720$. We set $C =2.0$.
The learning rate is set to $\lambda = 0.25$, and the SGD optimizer with a momentum of $0.9$ is used.
We consider $\nu$ from $0.01$ to $0.32$, which corresponds to a test accuracy between $65\%$ and $60\%$. 

Fig.~\ref{chap6-cifar10fig} illustrates the RDP and DP leakages for various methods. As shown, for this more challenging learning task, AdaScale still effectively reduces both RDP and DP leakages across different values of the convergence level $\nu$, and its performance is close to that of the optimal offline solution.

\begin{figure}[t]
    \centering
\includegraphics[width=0.98\linewidth]{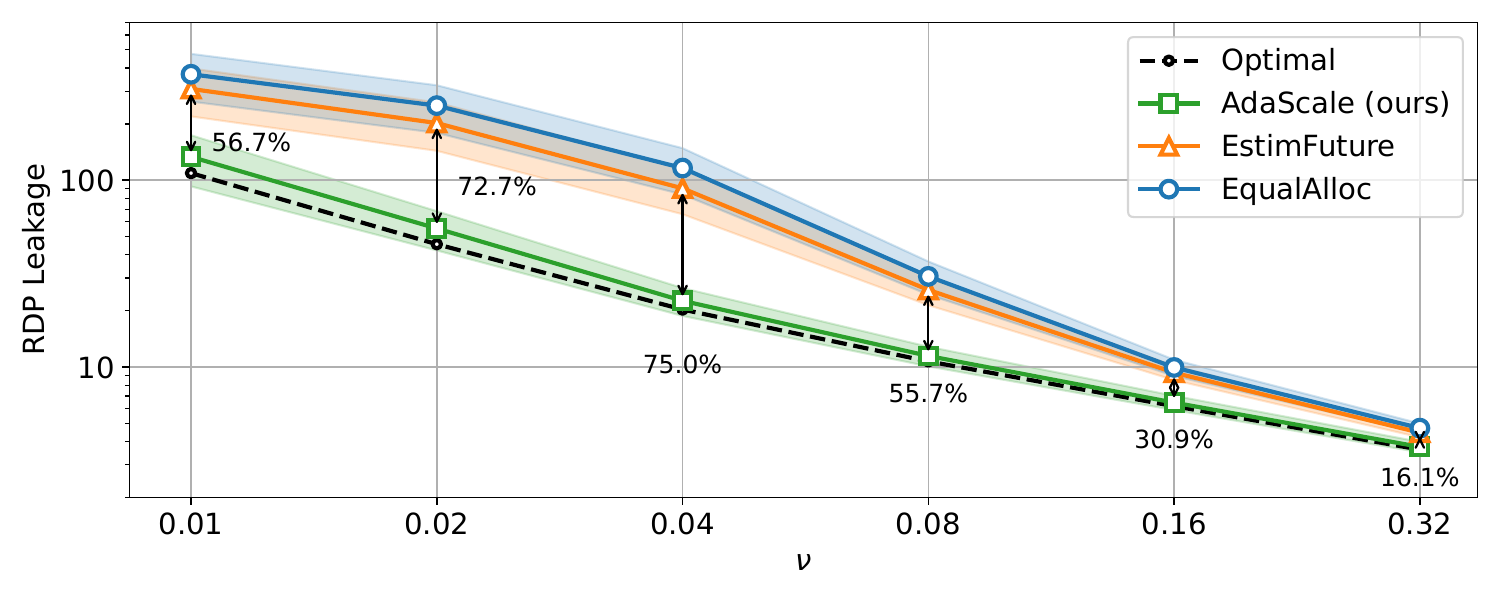}
\includegraphics[width=0.98\linewidth]{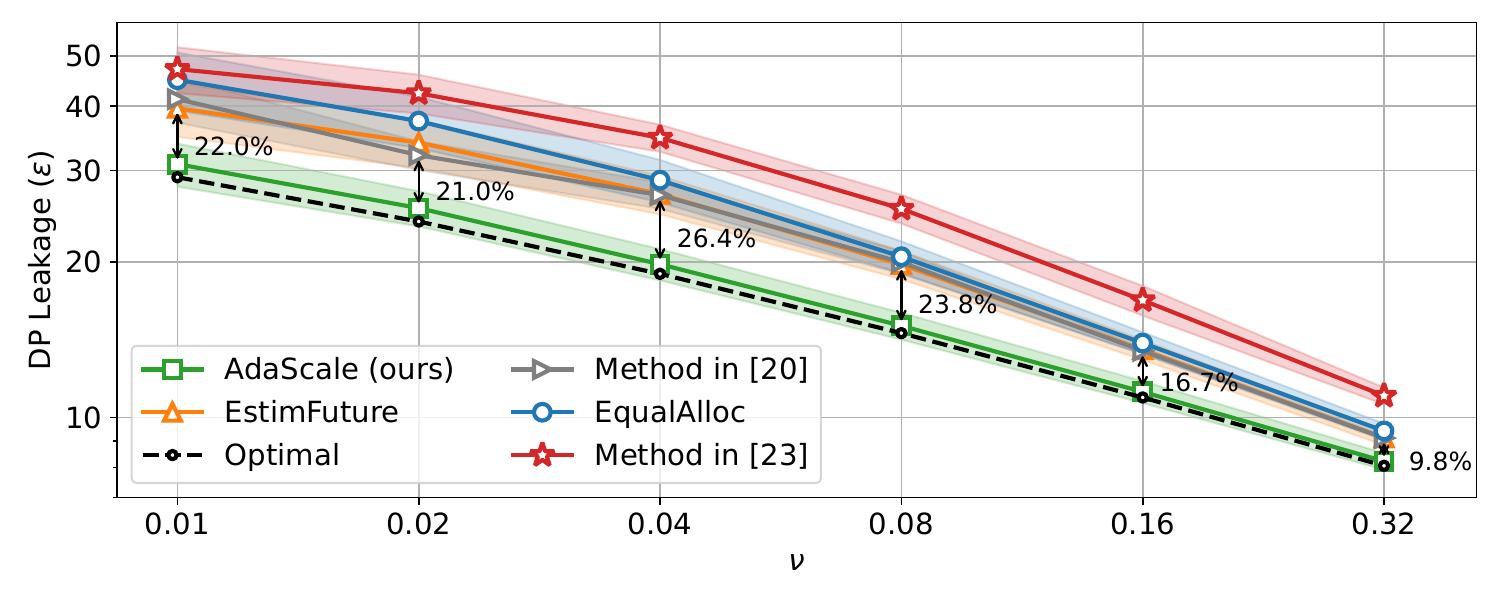}
\caption{RDP and DP leakage vs. $\nu$ for CIFAR-10. Range of $\nu$ corresponds to test accuracies between $60\%$ and $65\%$.}\label{chap6-cifar10fig}
\end{figure}

\section{Conclusion}\label{chapt6-summary}

In this work, we have investigated adaptive design of the receive scaling factors in an OTA FL system under dynamic wireless channel conditions, to reduce the overall privacy leakage during training. Unlike previous works, we aimed to minimize the overall RDP leakage directly while ensuring a specific level of convergence for the global loss function. 
We propose AdaScale, a novel online algorithm with per-round optimization problems that can be efficiently solved. Through novel bounding techniques, we derive upper bounds on the dynamic regret and constraint violation of the proposed algorithm, establishing that it achieves diminishing dynamic regret in time-averaged RDP leakage while ensuring convergence to a stationary point of the global loss function.
Numerical experiments show that our approach performs nearly optimally and effectively reduces both RDP and DP leakages compared with state-of-the-art benchmarks under the same learning performance.


\appendices 
\section{Proof of Theorem~\ref{chapt6-ConvergenceTheorem}}
\label{appendix-6A}

We first present the preliminary lemmas required for the proof in Appendix~\ref{chapt6-ExtraLemmas}, and then provide the complete proof of the theorem in Appendix~\ref{chapt6-ConvergenceTheoremProof}.

\subsection{Preliminary Lemmas for Proof of Theorem~\ref{chapt6-ConvergenceTheorem}}\label{chapt6-ExtraLemmas}

\begin{lemma}\label{chapt6-ExtraLemma1}
Suppose that assumption \textbf{A3} holds. Then, for the \( t \)-th round of the FedSGD algorithm described in Section~\ref{chapt6-FedSGDwithOTA}, the following equality holds:
\begin{align}
        \mathbb{E}& \Big[\Big< \nabla f(\mathbf{w}_t),  \mathbf{w}_{t+1}- \mathbf{w}_t  \Big> \Big|\mathbf{w}_t\Big] = -\lambda \| \nabla f(\mathbf{w}_t)\|^2.
\end{align}
\end{lemma}
\begin{proof}
Based on the model update in \eqref{ModelUpdate}, we have
\begin{align}
    &\mathbb{E} \Big[\Big< \nabla f(\mathbf{w}_t),  \mathbf{w}_{t+1}- \mathbf{w}_t  \Big> \Big|\mathbf{w}_t\Big] \nonumber \\ & \qquad =     \mathbb{E}\Big[\Big< \nabla f(\mathbf{w}_t), -\lambda \frac{\text{Re}(\mathbf{r}_t)}{\sqrt{\eta_t}}  \Big> \Big|\mathbf{w}_t\Big] \\ & \qquad \overset{(a)}{=}     \Big< \nabla f(\mathbf{w}_t), -\lambda \mathbb{E}\Big[\mathbf{s}_t +\tilde{\mathbf{n}}_t \Big | \mathbf{w}_t \Big] \Big> 
    \\ & \qquad \overset{(b)}{=}     \Big< \nabla f(\mathbf{w}_t), -\lambda \mathbb{E}\Big[\frac{1}{M} \sum_{m=1}^M {\mathbf{g}}_{m,t}  \Big | \mathbf{w}_t \Big] \Big> 
    \\ & \qquad \overset{(c)}{=}     \Big< \nabla f(\mathbf{w}_t), -\lambda \frac{1}{M} \sum_{m=1}^M \mathbb{E}\Big[\frac{\sum_{i=1 }^{n_m}{\mathbf{g}}_{m,t,i}}{n_m}  \Big | \mathbf{w}_t \Big] \Big> 
  \\ & \qquad \overset{(d)}{=}      \Big< \nabla f(\mathbf{w}_t), -\lambda \frac{1}{M} \sum_{m=1}^M \nabla f_m(\mathbf{w}_t)\Big> 
    \\ & \qquad \overset{(e)}{=}   -\lambda \| \nabla f(\mathbf{w}_t)\|^2,
\end{align}
where (a) follows the definitions of $\mathbf{s}_t$ and $\tilde{\mathbf{n}}_t $ in \eqref{decmpose}, (b) is due to the fact that $\tilde{\mathbf{n}}_t $ is zero-mean and independent of $\mathbf{w}_t$, (c) follows from $\mathbb{E}[{\mathbf{g}}_{m,t}| \mathbf{w}_t] = \frac{1}{B_m}\mathbb{E}\big[ \mathbb{E}_{\mathcal{B}_{m,t}}\big[\sum_{i \in \mathcal{B}_{m,t}}{\mathbf{g}}_{m,t,i}\big] \big| \mathbf{w}_t\big] = \frac{1}{B_m}\mathbb{E}[\sum_{i=1 }^{n_m} \frac{B_m}{n_m}{\mathbf{g}}_{m,t,i}| \mathbf{w}_t]$ due to Poisson sampling with rate $\frac{B}{n_m}$, (d) follows from \eqref{A3-Unbiasedness} in assumption \textbf{A3} , and finally (e) follows the definition of global loss function in \eqref{GlobalLoss}.
\end{proof}

\begin{lemma}\label{chapt6-ExtraLemma2}
Suppose that assumptions \textbf{A3} and \textbf{A4} hold. Then, for the \( t \)-th round of the FedSGD algorithm described in Section~\ref{chapt6-FedSGDwithOTA}, the following inequality holds:
\begin{align}
    \frac{L}{2} \mathbb{E}\Big[ &\| \mathbf{w}_{t+1}-  \mathbf{w}_t\|^2 \Big|\mathbf{w}_t\Big]  \le L \lambda^2 A_2   
     +\frac{L \lambda^2 d \sigma_n^2}{4 \eta_t} \nonumber \\ & + 
    2L \lambda^2 (A_1+1) \Big( (C_1+1) \|\nabla f(\mathbf{w}_t) \|^2 +C_2\Big). \label{Lem2-final}
\end{align}
\end{lemma}
\begin{proof}
We have
\begin{align}
     \frac{L}{2} \mathbb{E}\Big[ \| \mathbf{w}_{t+1}-  \mathbf{w}_t\|^2 \Big|\mathbf{w}_t\Big]  & \overset{(a)}{=}  \frac{L\lambda^2}{2} \mathbb{E}\Big[ \| \mathbf{s}_t + \tilde{\mathbf{n}}_t \|^2 \Big|\mathbf{w}_t\Big] \\ & \hspace{-1em} \overset{(b)}{=}
     \frac{L\lambda^2}{2} \mathbb{E}\Big[ \| \mathbf{s}_t \|^2+ \|\tilde{\mathbf{n}}_t \|^2 \Big|\mathbf{w}_t\Big] \\ & \hspace{-1em} \overset{(c)}{=} 
     \frac{L\lambda^2}{2} \Big( \mathbb{E}\big[ \| \mathbf{s}_t \|^2 \big| \mathbf{w}_t \big]+ \frac{d \sigma_n^2}{2 \eta_t} \Big),\label{Lem2-eq1}
 \end{align}
 where (a) follows the model update in \eqref{ModelUpdate}, (b) holds since $\tilde{\mathbf{n}}_t$ is zero-mean and independent of $\mathbf{s}_t$, and (c) follows by replacing the variance of $\tilde{\mathbf{n}}_t$ using  \eqref{decmpose}. 
 Now we proceed to bound the first term in \eqref{Lem2-eq1} as
\begin{align}
   & \mathbb{E} \big[ \| \mathbf{s}_t \|^2 \big| \mathbf{w}_t \big]   \overset{(a)}{=} \mathbb{E}\Big[\Big\|\frac{1}{M}\sum_{m=1}^M \frac{1}{B_m} \sum_{i \in \mathcal{B}_{m,t}} \mathbf{g}_{m,t,i} \Big\|^2 \Big | \mathbf{w}_t\Big] \\ & \overset{(b)}{=} \mathbb{E}\Big[\Big\|\frac{1}{M}\sum_{m=1}^M \frac{1}{B_m} \sum_{i \in \mathcal{B}_{m,t}} \Big( \nabla f_m(\mathbf{w}_t)+ \mathbf{z}_{m,t,i}\Big) \Big\|^2 \Big | \mathbf{w}_t\Big] \\ &  \overset{(c)}{=} \mathbb{E}\Big[\Big\|\frac{1}{M}\sum_{m=1}^M \frac{1}{B_m} \sum_{i \in \mathcal{B}_{m,t}}  \nabla f_m(\mathbf{w}_t) \Big\|^2 \Big | \mathbf{w}_t\Big] \nonumber  \\ & \qquad + \mathbb{E}\Big[\Big\|\frac{1}{M}\sum_{m=1}^M \frac{1}{B_m} \sum_{i \in \mathcal{B}_{m,t}} \mathbf{z}_{m,t,i} \Big\|^2 \Big | \mathbf{w}_t\Big],\label{Lem2-eq2}
\end{align}
where (a) follows from the definitions of $\mathbf{s}_t$ and $\mathbf{g}_{m,t}$ in \eqref{decmpose} and \eqref{averaging}, respectively; (b) follows from assumption \textbf{A3}; and (c) holds since \( \mathbf{z}_{m,t,i} \) is zero-mean based on assumption \textbf{A3}. 

Given \( \mathbf{w}_t \), the only source of randomness in the first term on the RHS of \eqref{Lem2-eq2} is the batch sampling, i.e., \( \mathcal{B}_{m,t} \). We can further upper bound this term as follows:
\begin{align}
    &\mathbb{E}_{\mathcal{B}_{m,t}}\Big[\Big\|\frac{1}{M}\sum_{m=1}^M \frac{1}{B_m} \sum_{i \in \mathcal{B}_{m,t}}  \nabla f_m(\mathbf{w}_t) \Big\|^2 \Big] \nonumber \\ &  \overset{(a)}{\le} \frac{1}{M}\sum_{m=1}^M \mathbb{E}_{\mathcal{B}_{m,t}}\Big[\frac{|\mathcal{B}_{m,t} |}{B_m^2} \sum_{i \in \mathcal{B}_{m,t}}  \|\nabla f_m(\mathbf{w}_t) \|^2 \Big] \\ & \overset{(b)}{=} \frac{1}{M}\sum_{m=1}^M \mathbb{E}_{\mathcal{B}_{m,t}}\Big[\frac{|\mathcal{B}_{m,t} |^2}{B_m^2}\Big]  \|\nabla f_m(\mathbf{w}_t) \|^2  \\ &  \overset{(c)}{\le} \frac{2}{M}\sum_{m=1}^M  \|\nabla f_m(\mathbf{w}_t) -\nabla f(\mathbf{w}_t)+  \nabla f(\mathbf{w}_t) \|^2 \\ &  \overset{(d)}{\le} \frac{4}{M}\sum_{m=1}^M \Big( \|\nabla f_m(\mathbf{w}_t) -\nabla f(\mathbf{w}_t)\|^2+  \|\nabla f(\mathbf{w}_t) \|^2 \Big)
    \\ &  \overset{(e)}{\le} 4 \Big( (C_1+1) \|\nabla f(\mathbf{w}_t) \|^2 +C_2\Big), \label{Lem2-eq3}
\end{align}
where (a) is derived by applying the inequality \( \left\| \sum_{j=1}^{J} \mathbf{y}_j \right\|^2 \le J \sum_{j=1}^{J} \| \mathbf{y}_j \|^2 \) to both summations over $m$ and $i$; (b) is derived by simplifying; (c) follows from the fact that \( \mathbb{E}[\frac{|\mathcal{B}_{m,t}|^2}{B_m^2}] = 1+  \frac{(1-q_m)}{B_m} \le 2 \), which holds under Poisson sampling with rate \( q_m = \frac{B_m}{n_m} \); (d)  holds by the inequality \( \| \mathbf{y}_1 + \mathbf{y}_2 \|^2 \le 2( \| \mathbf{y}_1 \|^2 + \| \mathbf{y}_2 \|^2 ) \);  
and (e) is derived using assumption \textbf{A4}.

The second term on the RHS of \eqref{Lem2-eq2}, can be upper bounded as
\begin{align}
\mathbb{E}&\Big[\Big\|\frac{1}{M}\sum_{m=1}^M \frac{1}{B_m} \sum_{i \in \mathcal{B}_{m,t}} \mathbf{z}_{m,t,i} \Big\|^2 \Big | \mathbf{w}_t\Big] \nonumber \\ & \hspace{-0.5em} \overset{(a)}{\le} \frac{1}{M}\sum_{m=1}^M \mathbb{E} \Big[\frac{|\mathcal{B}_{m,t}|}{B_m^2}\sum_{i \in \mathcal{B}_{m,t}} \|\mathbf{z}_{m,t,i}\|^2 \Big| \mathbf{w}_t\Big] \\ & \hspace{-0.5em} \overset{(b)}{=}  \frac{1}{M}\sum_{m=1}^M \mathbb{E}_{\mathcal{B}_{m,t}} \Big[\frac{|\mathcal{B}_{m,t}|}{B_m^2}\sum_{i \in \mathcal{B}_{m,t}} \mathbb{E}\big[\|\mathbf{z}_{m,t,i}\|^2 \big| \mathbf{w}_t\big]\Big] \\ & \hspace{-0.5em} \overset{(c)}{\le}  \frac{1}{M}\sum_{m=1}^M \mathbb{E}_{\mathcal{B}_{m,t}} \Big[\frac{|\mathcal{B}_{m,t}|^2}{B_m^2}\Big] \Big( A_1 \| \nabla f_m(\mathbf{w}_t) \|^2+ A_2\Big)  \\ & \hspace{-0.5em} \overset{(d)}{\le} \frac{2}{M}\sum_{m=1}^M \Big(A_1 \| \nabla f_m(\mathbf{w}_t)\|^2+A_2\Big) \\ & \hspace{-0.5em} \overset{(e)}{=} \frac{2A_1}{M}\sum_{m=1}^M  \| \nabla f_m(\mathbf{w}_t) -\nabla f(\mathbf{w}_t)+ \nabla f(\mathbf{w}_t)\|^2 + 2A_2
\\ & \hspace{-0.5em} \overset{(f)}{\le} \frac{4 A_1}{M}\sum_{m=1}^M  \| \nabla f_m(\mathbf{w}_t) - \nabla f(\mathbf{w}_t)\|^2 \nonumber \\ & \qquad \qquad + \frac{4A_1}{M}\sum_{m=1}^M  \| \nabla f(\mathbf{w}_t)\|^2+ 2A_2
\\ & \hspace{-0.5em} \overset{(g)}{\le}4A_1 \Big((C_1+1) \| \nabla f(\mathbf{w}_t\|^2+C_2 \Big) + 2A_2 , \label{Lem2-eq4}
\end{align}
where (a) is derived by applying the inequality \( \left\| \sum_{j=1}^{J} \mathbf{y}_j \right\|^2 \le J \sum_{j=1}^{J} \| \mathbf{y}_j \|^2 \) to both summations;  
(b) follows by decomposing the expectation over batch sampling and other sources of randomness in round \( t \);  
(c) follows from \eqref{BoundedVar} in \textbf{A3};  
(d) follows from the fact that  \( \mathbb{E}[\frac{|\mathcal{B}_{m,t}|^2}{B_m^2}] = 1+ \frac{(1-q_m)}{B_m} \le 2 \), which holds under Poisson sampling with rate \( q_m = \frac{B_m}{n_m} \);  
(e) follows directly by rearranging the terms;  
(f) holds by the inequality \( \| \mathbf{y}_1 + \mathbf{y}_2 \|^2 \le 2( \| \mathbf{y}_1 \|^2 + \| \mathbf{y}_2 \|^2 ) \);  
and (g) is derived by applying assumption \textbf{A4}.

Now, we substitute \eqref{Lem2-eq3} and \eqref{Lem2-eq4} in \eqref{Lem2-eq2} to form an upper bound on $\mathbb{E}[\|\mathbf{s}_t\|^2 | \mathbf{w}_t]$. Then, plugging in this upper bound in \eqref{Lem2-eq1}, we have
\begin{align}
    &\frac{L}{2} \mathbb{E}\Big[ \| \mathbf{w}_{t+1}-  \mathbf{w}_t\|^2 \Big|\mathbf{w}_t\Big]  \le 
      L \lambda^2 A_2  
     +\frac{L \lambda^2 d \sigma_n^2}{4 \eta_t} \nonumber \\ & + 2L \lambda^2 (A_1+1) \Big( (C_1+1) \|\nabla f(\mathbf{w}_t) \|^2 +C_2\Big)
     ,\label{Lem2-eq5}
\end{align}
which completes the proof.
\end{proof}

\subsection{Proof of Theorem~\ref{chapt6-ConvergenceTheorem}}\label{chapt6-ConvergenceTheoremProof}

\begin{proof}
Based on \textbf{A1}, we have
\begin{align}
    f(\mathbf{w}_{t+1}) &\le     f(\mathbf{w}_{t}) + \Big< \nabla f(\mathbf{w}_t), \mathbf{w}_{t+1}- \mathbf{w}_t  \Big> \nonumber \\ & \qquad \qquad + \frac{L}{2} \| \mathbf{w}_{t+1}- \mathbf{w}_t\|^2.\label{using A1}
\end{align}
Taking expectation from both sides of \eqref{using A1} on the randomness of round $t$ given $\mathbf{w}_t$, we have
\begin{align}
     &\mathbb{E}[f(\mathbf{w}_{t+1}) | \mathbf{w}_t] \nonumber \\ & \le     f(\mathbf{w}_{t}) +    \mathbb{E}\Big[\Big< \nabla f(\mathbf{w}_t), \mathbf{w}_{t+1}- \mathbf{w}_t  \Big> \Big|\mathbf{w}_t\Big] \nonumber \\ & \qquad +    \frac{L}{2} \mathbb{E}\Big[ \| \mathbf{w}_{t+1}- \mathbf{w}_t\|^2 \Big|\mathbf{w}_t\Big] \\ & \overset{(a)}{\le}
     f(\mathbf{w}_{t})  -\lambda \| \nabla f(\mathbf{w}_t)\|^2  + L \lambda^2 A_2   
     +\frac{L \lambda^2 d \sigma_n^2}{4 \eta_t} \nonumber \\ & \qquad +   2L \lambda^2(A_1+1) \Big( (C_1+1) \|\nabla f(\mathbf{w}_t) \|^2 +C_2\Big) 
     \\ & \overset{(b)}{=}
     f(\mathbf{w}_{t}) -  \lambda \Big(1- 2L \lambda (C_1+1)(A_1+1) \Big) \| \nabla f(\mathbf{w}_t)\|^2 \nonumber \\ & \qquad  + 
     L \lambda^2   
     \Big(2C_2(A_1+1) + A_2 \Big) \! 
     + \! \frac{L \lambda^2 d \sigma_n^2}{4 \eta_t}, \label{convt-eq1}
\end{align}
where (a) results from Lemmas~\ref{chapt6-ExtraLemma1} and~\ref{chapt6-ExtraLemma2};  
and (b) is derived by rearranging the terms. Now, we take the expectation over all sources of randomness in the algorithm on both sides of \eqref{convt-eq1}. Rearranging the terms, we obtain
\begin{align}
   & \lambda  \Big(1-2L  \lambda (C_1+1)(A_1+1) \Big)\mathbb{E}  \| \nabla f(\mathbf{w}_t)\|^2 \le \frac{L \lambda^2 d \sigma_n^2}{4 \eta_t} \nonumber \\ & + \mathbb{E}[f(\mathbf{w}_t) - f(\mathbf{w}_{t+1})]  +   L \lambda^2 
     \Big(2C_2(A_1+1) + A_2 \Big)   
     . \label{convt-eq2}
\end{align}
Now, to simplify~\eqref{convt-eq2}, we set learning rate $\lambda$ such that 
\begin{align}
    1-2L \lambda (C_1+1)(A_1+1) \ge \frac{1}{2}. 
\end{align}
Thus, \eqref{convt-eq2} implies the following:
\begin{align}
    \frac{\lambda}{2} & \mathbb{E} \| \nabla f(\mathbf{w}_t)\|^2    \le \mathbb{E}[f(\mathbf{w}_t) - f(\mathbf{w}_{t+1})]  
     +\frac{L \lambda^2 d \sigma_n^2}{4 \eta_t} \nonumber \\ & \qquad + L \lambda^2   
     \Big(2C_2(A_1+1) + A_2 \Big).  \label{convt-eq3}
\end{align}
Summing both sides of \eqref{convt-eq3} from \( t = 0 \) to \( T - 1 \) and dividing by \( \frac{\lambda T}{2} \), we have
\begin{align}
    & \frac{1}{T}  \sum_{t=0}^{T-1}\mathbb{E} \| \nabla f(\mathbf{w}_t)\|^2    \le \frac{2\big(f(\mathbf{w}_0) - \mathbb{E}[f(\mathbf{w}_{T})] \big)}{\lambda T}  \nonumber \\ &  \quad +   2L \lambda   
     \Big(2C_2(A_1+1) + A_2 \Big)   
     + \frac{L \lambda}{T} \sum_{t=0}^{T-1}\frac{ d \sigma_n^2}{2 \eta_t}. \label{convt-eq4}
\end{align}
Further upper bounding \( f(\mathbf{w}_0) - \mathbb{E}[f(\mathbf{w}_{T})] \) on the RHS of \eqref{convt-eq4} by \( f(\mathbf{w}_0) - f^\star \) using assumption \textbf{A2} yields the result stated in Theorem~\ref{chapt6-ConvergenceTheorem}.
\end{proof}

\bibliographystyle{IEEEtran}
\bibliography{refs}

\end{document}